\documentclass[journal]{IEEEtran}
\IEEEoverridecommandlockouts

\UseRawInputEncoding
\def\BibTeX{{\rm B\kern-.05em{\sc i\kern-.025em b}\kern-.08em T\kern-.1667em\lower.7ex\hbox{E}\kern-.125emX}}
\usepackage{cite}
\usepackage{amsmath,amssymb,amsfonts}
\usepackage{algorithmic}
\usepackage{graphicx}
\usepackage{textcomp}
\usepackage{xcolor}
\usepackage{amsthm}
\usepackage{amsmath}
\usepackage{amssymb}
\usepackage{booktabs}
\usepackage{amsfonts}
\usepackage{lipsum}
\usepackage[inkscapeformat=png]{svg}
\usepackage{mathtools}
\usepackage{threeparttable}
\usepackage{tablefootnote}
\usepackage[hyphens]{url}
\usepackage[hidelinks]{hyperref}
\usepackage{algorithmic, algorithm}
\usepackage{tabularx}
\usepackage{subfig}
\usepackage{mwe}
\usepackage{ragged2e}
\usepackage{dblfloatfix}
\usepackage{footmisc}
\usepackage{footnote}
\usepackage{comment}
\usepackage[labelfont=bf,font=small]{caption}
\makesavenoteenv{tabular}
\makesavenoteenv{table}

\newtheorem{lem}{Lemma}

\usepackage[inline]{enumitem}

\makeatletter
\newcommand\fs@betterruled{%
  \def\@fs@cfont{\bfseries}\let\@fs@capt\floatc@ruled
  \def\@fs@pre{\vspace*{5pt}\hrule height.8pt depth0pt \kern2pt}%
  \def\@fs@post{\kern2pt\hrule\relax}%
  \def\@fs@mid{\kern2pt\hrule\kern2pt}%
  \let\@fs@iftopcapt\iftrue}
\floatstyle{betterruled}
\restylefloat{algorithm}
\makeatother

\definecolor{custom_blue}{HTML}{1F77B4}
\definecolor{custom_orange}{HTML}{FF7F0E}
\definecolor{custom_green}{HTML}{2CA02C}

\usepackage{tikz}
\usepackage{pgfplots}
\pgfplotsset{compat=1.17}

\newcommand{\dg}[1]{\textcolor{red}{{#1}}}
\newcommand{\gf}[1]{\textcolor{cyan}{{#1}}}

\newcommand{\mx}[1]{\mathbf{#1}}

\usepackage[nolist,printonlyused]{acronym}
\begin{document}

\begin{acronym}
  \acro{2G}{Second Generation}
  \acro{3G}{3$^\text{rd}$~Generation}
  \acro{3GPP}{3$^\text{rd}$~Generation Partnership Project}
  \acro{4G}{4$^\text{th}$~Generation}
  \acro{5G}{5$^\text{th}$~Generation}
  \acro{AA}{Antenna Array}
  \acro{AC}{Admission Control}
  \acro{AD}{Attack-Decay}
  \acro{ADSL}{Asymmetric Digital Subscriber Line}
	\acro{AHW}{Alternate Hop-and-Wait}
  \acro{AMC}{Adaptive Modulation and Coding}
	\acro{AP}{access point}
  \acro{AoA}{angle of arrival}
  \acro{AoI}{age of information}
  \acro{APA}{Adaptive Power Allocation}
  \acro{AR}{autoregressive}
  \acro{ARMA}{Autoregressive Moving Average}
  \acro{ATES}{Adaptive Throughput-based Efficiency-Satisfaction Trade-Off}
  \acro{AWGN}{additive white Gaussian noise}
  \acro{BB}{Branch and Bound}
  \acro{BD}{Block Diagonalization}
  \acro{BER}{bit error rate}
  \acro{BF}{Best Fit}
  \acro{BLER}{BLock Error Rate}
  \acro{BPC}{Binary power control}
  \acro{BPSK}{binary phase-shift keying}
  \acro{BPA}{Best \ac{PDPR} Algorithm}
  \acro{BRA}{Balanced Random Allocation}
  \acro{BS}{base station}
  \acro{CAP}{Combinatorial Allocation Problem}
  \acro{CAPEX}{Capital Expenditure}
  \acro{CAZAC}{Constant Amplitude Zero Autocorrelation}
  \acro{CB}{codebook}
  \acro{CBF}{Coordinated Beamforming}
  \acro{CBR}{Constant Bit Rate}
  \acro{CBS}{Class Based Scheduling}
  \acro{CC}{Congestion Control}
  \acro{CDF}{Cumulative Distribution Function}
  \acro{CDMA}{Code-Division Multiple Access}
  \acro{CL}{Closed Loop}
  \acro{CLPC}{Closed Loop Power Control}
  \acro{CNR}{Channel-to-Noise Ratio}
  \acro{CPA}{Cellular Protection Algorithm}
  \acro{CPICH}{Common Pilot Channel}
  \acro{CoMP}{Coordinated Multi-Point}
  \acro{CQI}{Channel Quality Indicator}
  \acro{CRLB}{Cram\'er-Rao Lower Bound}
  \acro{CRM}{Constrained Rate Maximization}
	\acro{CRN}{Cognitive Radio Network}
  \acro{CS}{Coordinated Scheduling}
  \acro{CSI}{channel state information}
  \acro{CSIR}{channel state information at the receiver}
  \acro{CSIT}{channel state information at the transmitter}
  \acro{CUE}{cellular user equipment}
  \acro{D2D}{device-to-device}
  \acro{DCA}{Dynamic Channel Allocation}
  \acro{DE}{Differential Evolution}
  \acro{DFT}{Discrete Fourier Transform}
  \acro{DISCOVER}{Deep Intrinsically Motivated Exploration}
  \acro{DIST}{Distance}
  \acro{DL}{downlink}
  \acro{DMA}{Double Moving Average}
	\acro{DMRS}{Demodulation Reference Signal}
  \acro{D2DM}{D2D Mode}
  \acro{DMS}{D2D Mode Selection}
  \acro{DPC}{Dirty Paper Coding}
  \acro{DRA}{Dynamic Resource Assignment}
  \acro{DRL}{deep reinforcement learning}
  \acro{DDPG}{Deep Deterministic Policy Gradient}
  \acro{RL}{reinforcement learning}
  \acro{SAC}{Soft Actor-Critic}
  \acro{DSA}{Dynamic Spectrum Access}
  \acro{DSM}{Delay-based Satisfaction Maximization}
  \acro{ECC}{Electronic Communications Committee}
  \acro{ECRB}{expectation of the conditional Cramér-Rao bound}
  \acro{EFLC}{Error Feedback Based Load Control}
  \acro{EI}{Efficiency Indicator}
  \acro{eNB}{Evolved Node B}
  \acro{EPA}{Equal Power Allocation}
  \acro{EPC}{Evolved Packet Core}
  \acro{EPS}{Evolved Packet System}
  \acro{E-UTRAN}{Evolved Universal Terrestrial Radio Access Network}
  \acro{ES}{Exhaustive Search}
  \acro{FDD}{frequency division duplexing}
  \acro{FDM}{Frequency Division Multiplexing}
  \acro{FER}{Frame Erasure Rate}
  \acro{FF}{Fast Fading}
  \acro{FIM}{Fisher Information Matrix}
  \acro{FSB}{Fixed Switched Beamforming}
  \acro{FST}{Fixed SNR Target}
  \acro{FTP}{File Transfer Protocol}
  \acro{GA}{Genetic Algorithm}
  \acro{GBR}{Guaranteed Bit Rate}
  \acro{GLR}{Gain to Leakage Ratio}
  \acro{GOS}{Generated Orthogonal Sequence}
  \acro{GPL}{GNU General Public License}
  \acro{GRP}{Grouping}
  \acro{HARQ}{Hybrid Automatic Repeat Request}
  \acro{HMS}{Harmonic Mode Selection}
  \acro{HOL}{Head Of Line}
  \acro{HSDPA}{High-Speed Downlink Packet Access}
  \acro{HSPA}{High Speed Packet Access}
  \acro{HTTP}{HyperText Transfer Protocol}
  \acro{ICMP}{Internet Control Message Protocol}
  \acro{ICI}{Intercell Interference}
  \acro{ID}{Identification}
  \acro{IETF}{Internet Engineering Task Force}
  \acro{ILP}{Integer Linear Program}
  \acro{JRAPAP}{Joint RB Assignment and Power Allocation Problem}
  \acro{UID}{Unique Identification}
  \acro{i.i.d.}{independent and identically distributed}
  \acro{IIR}{Infinite Impulse Response}
  \acro{ILP}{Integer Linear Problem}
  \acro{IMT}{International Mobile Telecommunications}
  \acro{INV}{Inverted Norm-based Grouping}
	\acro{IoT}{Internet of Things}
  \acro{IP}{Internet Protocol}
  \acro{IPv6}{Internet Protocol Version 6}
  \acro{IRS}{intelligent reflecting surface}
  \acro{ISD}{Inter-Site Distance}
  \acro{ISI}{Inter Symbol Interference}
  \acro{ITU}{International Telecommunication Union}
  \acro{JOAS}{Joint Opportunistic Assignment and Scheduling}
  \acro{JOS}{Joint Opportunistic Scheduling}
  \acro{JP}{Joint Processing}
	\acro{JS}{Jump-Stay}
    \acro{KF}{Kalman filter}
  \acro{KKT}{Karush-Kuhn-Tucker}
  \acro{L3}{Layer-3}
  \acro{LAC}{Link Admission Control}
  \acro{LA}{Link Adaptation}
  \acro{LC}{Load Control}
  \acro{LMMSE}{linear minimum mean squared error}
  \acro{LOS}{line of sight}
  \acro{LP}{Linear Programming}
  \acro{LS}{least squares}
  \acro{LTE}{Long Term Evolution}
  \acro{LTE-A}{LTE-Advanced}
  \acro{LTE-Advanced}{Long Term Evolution Advanced}
  \acro{M2M}{Machine-to-Machine}
  \acro{MAC}{Medium Access Control}
  \acro{MANET}{Mobile Ad hoc Network}
  \acro{MC}{Modular Clock}
  \acro{MCS}{Modulation and Coding Scheme}
  \acro{MDB}{Measured Delay Based}
  \acro{MDI}{Minimum D2D Interference}
  \acro{MF}{Matched Filter}
  \acro{MG}{Maximum Gain}
  \acro{MH}{Multi-Hop}
  \acro{MIMO}{multiple-input multiple-output}
  \acro{MINLP}{Mixed Integer Nonlinear Programming}
  \acro{MIP}{Mixed Integer Programming}
  \acro{MISO}{Multiple-input single-output}
  \acro{ML}{maximum likelihood}
  \acro{MLWDF}{Modified Largest Weighted Delay First}
  \acro{MME}{Mobility Management Entity}
  \acro{MML}{misspecified maximum likelihood}
  \acro{MMSE}{minimum mean squared error}
  \acro{MOS}{Mean Opinion Score}
  \acro{MPF}{Multicarrier Proportional Fair}
  \acro{MRA}{Maximum Rate Allocation}
  \acro{MR}{Maximum Rate}
  \acro{MRC}{maximum ratio combining}
  \acro{MRT}{Maximum Ratio Transmission}
  \acro{MRUS}{Maximum Rate with User Satisfaction}
  \acro{MS}{mobile station}
  \acro{MSE}{mean squared error}
  \acro{MSI}{Multi-Stream Interference}
  \acro{MTC}{Machine-Type Communication}
  \acro{MTSI}{Multimedia Telephony Services over IMS}
  \acro{MTSM}{Modified Throughput-based Satisfaction Maximization}
  \acro{MU-MIMO}{multiuser multiple-input multiple-output}
  \acro{MU-MISO}{multiuser multiple-input single-output}
  \acro{MU}{multi-user}
  \acro{NAS}{Non-Access Stratum}
  \acro{NB}{Node B}
  \acro{NE}{Nash equilibrium}
  \acro{NCL}{Neighbor Cell List}
  \acro{NLP}{Nonlinear Programming}
  \acro{NLOS}{Non-Line of Sight}
  \acro{NMSE}{normalized mean squared error}
  \acro{NOMA}{non-orthogonal multiple access}
  \acro{NORM}{Normalized Projection-based Grouping}
  \acro{NP}{Non-Polynomial Time}
  \acro{NR}{New Radio}
  \acro{NRT}{Non-Real Time}
  \acro{NSPS}{National Security and Public Safety Services}
  \acro{O2I}{Outdoor to Indoor}
  \acro{OFDMA}{orthogonal frequency division multiple access}
  \acro{OFDM}{orthogonal frequency division multiplexing}
  \acro{OFPC}{Open Loop with Fractional Path Loss Compensation}
	\acro{O2I}{Outdoor-to-Indoor}
  \acro{OL}{Open Loop}
  \acro{OLPC}{Open-Loop Power Control}
  \acro{OL-PC}{Open-Loop Power Control}
  \acro{OPEX}{Operational Expenditure}
  \acro{ORB}{Orthogonal Random Beamforming}
  \acro{JO-PF}{Joint Opportunistic Proportional Fair}
  \acro{OSI}{Open Systems Interconnection}
  \acro{PAIR}{D2D Pair Gain-based Grouping}
  \acro{PAPR}{Peak-to-Average Power Ratio}
  \acro{P2P}{Peer-to-Peer}
  \acro{PC}{Power Control}
  \acro{PCI}{Physical Cell ID}
  \acro{PDF}{Probability Density Function}
  \acro{PDPR}{pilot-to-data power ratio}
  \acro{PER}{Packet Error Rate}
  \acro{PF}{Proportional Fair}
  \acro{P-GW}{Packet Data Network Gateway}
  \acro{PL}{Pathloss}
  \acro{PPR}{pilot power ratio}
  \acro{PRB}{physical resource block}
  \acro{PROJ}{Projection-based Grouping}
  \acro{ProSe}{Proximity Services}
  \acro{PS}{Packet Scheduling}
  \acro{PSAM}{pilot symbol assisted modulation}
  \acro{PSO}{Particle Swarm Optimization}
  \acro{PZF}{Projected Zero-Forcing}
  \acro{QAM}{Quadrature Amplitude Modulation}
  \acro{QoS}{Quality of Service}
  \acro{QPSK}{Quadri-Phase Shift Keying}
  \acro{RAISES}{Reallocation-based Assignment for Improved Spectral Efficiency and Satisfaction}
  \acro{RAN}{Radio Access Network}
  \acro{RA}{Resource Allocation}
  \acro{RAT}{Radio Access Technology}
  \acro{RATE}{Rate-based}
  \acro{RB}{resource block}
  \acro{RBG}{Resource Block Group}
  \acro{REF}{Reference Grouping}
  \acro{RIS}{reconfigurable intelligent surface}
  \acro{RLC}{Radio Link Control}
  \acro{RM}{Rate Maximization}
  \acro{RNC}{Radio Network Controller}
  \acro{RND}{Random Grouping}
  \acro{RRA}{Radio Resource Allocation}
  \acro{RRM}{Radio Resource Management}
  \acro{RSCP}{Received Signal Code Power}
  \acro{RSRP}{Reference Signal Receive Power}
  \acro{RSRQ}{Reference Signal Receive Quality}
  \acro{RR}{Round Robin}
  \acro{RRC}{Radio Resource Control}
  \acro{RSSI}{Received Signal Strength Indicator}
  \acro{RT}{Real Time}
  \acro{RU}{Resource Unit}
  \acro{RUNE}{RUdimentary Network Emulator}
  \acro{RV}{Random Variable}
  \acro{SAC}{Soft Actor-Critic}
  \acro{SCM}{Spatial Channel Model}
  \acro{SC-FDMA}{Single Carrier - Frequency Division Multiple Access}
  \acro{SD}{Soft Dropping}
  \acro{S-D}{Source-Destination}
  \acro{SDPC}{Soft Dropping Power Control}
  \acro{SDMA}{Space-Division Multiple Access}
  \acro{SE}{spectral efficiency}
  \acro{SER}{Symbol Error Rate}
  \acro{SES}{Simple Exponential Smoothing}
  \acro{S-GW}{Serving Gateway}
  \acro{SINR}{signal-to-interference-plus-noise ratio}
  \acro{SI}{Satisfaction Indicator}
  \acro{SIP}{Session Initiation Protocol}
  \acro{SISO}{single input single output}
  \acro{SIMO}{single input multiple output}
  \acro{SIR}{signal-to-interference ratio}
  \acro{SLNR}{Signal-to-Leakage-plus-Noise Ratio}
  \acro{SMA}{Simple Moving Average}
  \acro{SNR}{signal-to-noise ratio}
  \acro{SORA}{Satisfaction Oriented Resource Allocation}
  \acro{SORA-NRT}{Satisfaction-Oriented Resource Allocation for Non-Real Time Services}
  \acro{SORA-RT}{Satisfaction-Oriented Resource Allocation for Real Time Services}
  \acro{SPF}{Single-Carrier Proportional Fair}
  \acro{SRA}{Sequential Removal Algorithm}
  \acro{SRS}{Sounding Reference Signal}
  \acro{SU-MIMO}{single-user multiple input multiple output}
  \acro{SU}{Single-User}
  \acro{SVD}{Singular Value Decomposition}
  \acro{TCP}{Transmission Control Protocol}
  \acro{TDD}{time division duplexing}
  \acro{TDMA}{Time Division Multiple Access}
  \acro{TETRA}{Terrestrial Trunked Radio}
  \acro{TP}{Transmit Power}
  \acro{TPC}{Transmit Power Control}
  \acro{TTI}{Transmission Time Interval}
  \acro{TTR}{Time-To-Rendezvous}
  \acro{TSM}{Throughput-based Satisfaction Maximization}
  \acro{TU}{Typical Urban}
  \acro{UAV}{unmanned aerial vehicle}
  \acro{UE}{user equipment}
  \acro{UEPS}{Urgency and Efficiency-based Packet Scheduling}
  \acro{UL}{uplink}
  \acro{ULA}{uniform linear array}
  \acro{UMTS}{Universal Mobile Telecommunications System}
  \acro{URA}{uniform rectangular array}
  \acro{URI}{Uniform Resource Identifier}
  \acro{URM}{Unconstrained Rate Maximization}
  \acro{UT}{user terminal}
  \acro{VAR}{vector autoregressive}
  \acro{VR}{Virtual Resource}
  \acro{VoIP}{Voice over IP}
  \acro{WAN}{Wireless Access Network}
  \acro{WCDMA}{Wideband Code Division Multiple Access}
  \acro{WF}{Water-filling}
  \acro{WiMAX}{Worldwide Interoperability for Microwave Access}
  \acro{WINNER}{Wireless World Initiative New Radio}
  \acro{WLAN}{Wireless Local Area Network}
  \acro{WMPF}{Weighted Multicarrier Proportional Fair}
  \acro{WPF}{Weighted Proportional Fair}
  \acro{WSN}{Wireless Sensor Network}
  \acro{WSS}{wide-sense stationary}
  \acro{WWW}{World Wide Web}
  \acro{XIXO}{(Single or Multiple) Input (Single or Multiple) Output}
  \acro{ZF}{zero-forcing}
  \acro{ZMCSCG}{Zero Mean Circularly Symmetric Complex Gaussian}
\end{acronym}

\title{Combating Inter-Operator Pilot Contamination in Reconfigurable Intelligent Surfaces Assisted Multi-Operator Networks}

\author{Do\u{g}a G\"{u}rg\"{u}no\u{g}lu, \emph{Student Member, IEEE},
        Emil~Bj\"{o}rnson,~\emph{Fellow, IEEE},~G\'abor~Fodor, \emph{Senior Member, IEEE}
\thanks{D. G\"urg\"unoglu, E. Bj\"{o}rnson and G. Fodor are with the Faculty of Electrical Engineering and Computer Science, KTH Royal Institute of Technology, Stockholm 100 44, Sweden (e-mails: \{dogag,emilbjo,gaborf\}@kth.se). G. Fodor is also with Ericsson Research, Stockholm 164 80, Sweden (e-mail: \{gabor.fodor\}@ericsson.com). This study is supported by EU Horizon 2020 MSCA-ITN-METAWIRELESS, Grant Agreement 956256. E.~Bj\"ornson is funded by the FFL18-0277 grant from SSF. The preliminary version  \cite{Gurgunoglu2023_BlackSeaCom} of this work was presented at IEEE BlackSeacom 2023, Istanbul, Turkiye.}
}

\vspace{-0.8cm}
\vspace{-0.1cm}

\maketitle

\begin{abstract}
In this paper, we study a new kind of pilot contamination appearing in multi-operator reconfigurable intelligent surfaces (RIS) 
assisted networks, where multiple operators provide services to their respective served users. 
The operators use dedicated frequency bands, but each RIS inadvertently reflects the transmitted uplink signals of 
the user equipment devices in multiple bands. 
Consequently, the concurrent reflection of pilot signals during the channel estimation phase introduces a new inter-operator pilot contamination effect. 
We investigate the implications of this effect in systems with either deterministic or correlated Rayleigh fading channels, 
specifically focusing on its impact on channel estimation quality, signal equalization, and channel capacity. 
The numerical results demonstrate the substantial degradation in system performance caused by this phenomenon and highlight the pressing need 
to address inter-operator pilot contamination in multi-operator RIS deployments. 
To combat the negative effect of this new type of pilot contamination, we propose to use orthogonal RIS configurations during
uplink pilot transmission, which can mitigate or eliminate the negative effect of inter-operator pilot contamination at the
expense of some inter-operator information exchange and orchestration. 
\end{abstract}


\section{Introduction}
Pilot contamination is a key problem that frequently arises in wireless communication systems \cite{Marzetta2010a}. 
When multiple users use the same pilot sequences simultaneously in the same band, due to the limited channel coherence time, 
the \ac{BS} cannot distinguish their channels. 
This typically results in poor channel estimates and extra beamformed interference from or towards the \acp{UE} that reuse the same pilot sequence. 
Therefore, pilot contamination adversely affects the coherent reception of data, 
and methods to mitigate pilot contamination---including adaptive pilot reuse, power control, user grouping, multi-cell coordination, and coded random access techniques---have been widely studied in the communication literature \cite{Marzetta2010a,Sanguinetti2020a,pCon1,Saxena:15, Fodor:17}.

In recent years, \acp{RIS} have arisen as a new technology component for 6G \cite{ris_commag}. 
An \ac{RIS} is a surface consisting of multiple reflecting elements that have sub-wavelength spacing and controllable 
reflection properties \cite{Bjornson2020a}.
This feature provides partial control of
the propagation environment that can lead to better services for users, especially when their serving \ac{BS} is not in their \ac{LOS}. 
By adjusting the impedances of the individual elements via a \ac{RIS} controller, the elements are capable of adding desired phase shifts 
to the reflected signals, thereby forming reflected beams in desired directions that can significantly boost the \ac{SINR} 
and reduce the symbol estimation error at the receiver \cite{ris_commag, Araujo:23}.

On the other hand, the addition of \acp{RIS} to existing systems introduces new design and operational challenges \cite{Liu:23}. 
For example, the length of the pilot signal required by a single \ac{UE} is proportional to the number of \ac{RIS} elements 
(e.g., tens or hundreds), because the \ac{RIS} must change its configuration to explore all channel dimensions \cite{risChanEst,Bjornson2022b}.
In addition, the path loss of the reflected path through a passive \ac{RIS} is proportional to the multiplication of the path losses 
to and from the \ac{RIS} \cite{ris_zappone}, so a larger surface is needed to achieve a decent \ac{SNR} improvement. 
Active \acp{RIS}, on the other hand, use amplifiers to overcome the large path loss 
but have the traditional issues of relays: increased power consumption, higher cost, and additional noise \cite{Zhang2023,Rihan:23}. 
While the aforementioned problems caused by the \ac{RIS} have been recognized \cite{Garg:22}, pilot contamination 
caused by the presence of multiple \acp{RIS} has not been studied in the literature. 

Wireless communication systems use standardized protocols, interfaces, and well-defined pilot sequences 
and codebooks to ensure inter-operability \cite{38211}. 
While employing \acp{RIS} in cellular networks have not been studied by the relevant standards organizations yet, 
it may be expected that the configuration sequences that facilitate the deployment of \acp{RIS} while maintaining interoperability 
will be specified. 
Consequently, when multiple cellular networks are deployed by different network operators in overlapping geographical areas, 
the \acp{RIS} may adopt identical or overlapping pilot sequences and cause pilot contamination. The number of orthogonal pilot sequences is limited by the length of the pilot sequence, and increasing the pilot sequence length not only creates more channel estimation overhead but also is infeasible due to the limited coherence budget of the channel. As a consequence, the need for repeating pilot sequences comes up very often.

In this paper, we argue that 
when multiple \acp{RIS} are deployed for the purpose of shaping the propagation characteristics of the environment, 
the propagation characteristics might change in unintended ways. 
For example, an \ac{RIS} belonging to another operator might modify the propagation of a \ac{UE}'s own pilot signal, 
leading to pilot contamination even in the absence of any interfering signals or intra-band pilot reuse. 
The underlying reason is that an \ac{RIS} element---although designed for a particular frequency---does not act 
as a bandpass filter, but reflects all frequencies with varying amplitude and phase. Indeed, as pointed out in \cite{Zhang:20}, due to the lack of baseband signal processing, the 
\ac{RIS} reflects the impinging broadband signal with frequency-flat reflection coefficients.
Therefore, in realistic system models of, for example, passive \ac{RIS} assisted 5G New Radio systems, we
need to take into account that the \ac{RIS} inadvertently reflects the transmitted uplink signals of the user equipment devices in multiple bands as in \cite{Pei:21}, \cite{Tampio:21}.

Specifically,
in this paper, we identify this new pilot contamination phenomenon as a major practical challenge when multiple \ac{RIS} assisted operator networks
are deployed over a geographical area, including the important practical scenario of inter-operator site sharing \cite{Vincenzi:17,Chien:18}.
In such an environment, a \ac{UE} that transmits pilots to its serving \ac{BS} via multiple \ac{RIS}, which may change
their configurations simultaneously, is exposed to new pilot-related ambiguities that have not been studied before.
Since this phenomenon exacerbates the pilot contamination problem, it is clear that pilot contamination due 
to the presence of multiple \acp{RIS} must be dealt with.

To the best of the authors' knowledge, the problem of inter-operator pilot contamination 
in \ac{RIS}-aided wireless communication systems  has not been addressed before, except in the preliminary version of this manuscript \cite{Gurgunoglu2023_BlackSeaCom},
which assumed deterministic rather than stochastically fading channels.
In this paper, we derive the capacity lower bound of the system under pilot contamination and imperfect \ac{CSI} assuming Rayleigh fading.
Our major contributions can be summarized as follows:
\begin{itemize}
    \item For the case when inter-operator pilot contamination is neglected, we provide a misspecified \ac{ML} 
    estimator under the assumption that all the channels in the system setup are deterministic. 
    We also derive the resulting channel estimation \ac{MSE} for different choices of the \ac{RIS} configurations.
    \item 
    Based on the obtained results for the channel estimation error under inter-operator pilot contamination, 
    we provide the data signal estimation \ac{MSE} for a misspecified \ac{MMSE} estimator.
    \item 
    In addition to deterministic channels, we also consider the case where the channels 
    are Rayleigh fading with spatial correlation. 
    For generic channel spatial covariances, we derive the misspecified \ac{MMSE} estimator and the resulting \ac{MSE}.
    \item 
    Based on the channel estimation error model, 
    obtained for spatially correlated Rayleigh fading channels, 
    we derive a capacity lower bound given the channel estimates. 
    Our numerical results show that the choice of \ac{RIS} configurations 
    during channel estimation makes a significant impact on the capacity lower bound.
\end{itemize}
The rest of the manuscript is organized as follows: in Section \ref{sec:system_model}, we provide the received signal model, in Section \ref{sec:MLE}, we provide the misspecified \ac{ML} estimator where the inter-operator pilot contamination is neglected and the channels are assumed to be deterministic. Section \ref{sec:DataTransmission} builds on top of section \ref{sec:MLE} by providing the data estimation \ac{MSE} as a result of inter-operator pilot contamination. In Section \ref{sec:RayleighChan}, we derive the impact of inter-operator pilot contamination in closed-form for the case of spatially correlated Rayleigh-fading channels by considering the channel estimation \ac{MSE} as our performance metric. Since the data estimation \ac{MSE} has a dependence on individual channel realizations, an alternative performance metric for data transmission is needed to capture the behavior of fading channels. To this end, we derive the capacity lower bound under channel side information in Section \ref{sec:capBound}. We provide the numerical results in Section \ref{sec:NumRes}, and conclude the manuscript in Section \ref{sec:Conclusion}.
\section{System Model}\label{sec:system_model}
In this paper, we study the pilot contamination caused by the presence of multiple \acp{RIS} by considering the uplink of a system consisting of two wide-band \acp{RIS}, two single-antenna \acp{UE}, and two co-located single-antenna \acp{BS}. The \acp{UE} are subscribed to different operators and use non-overlapping frequency bands. Each \ac{RIS} is dedicated to and controlled by a single operator, but both \ac{UE} signals are reflected from both \acp{RIS}. In this scenario, although there is no interference between the two \acp{UE}, both \acp{RIS} affect both frequency bands.

In Fig.~\ref{fig:systemSetup}, we graphically describe the system that we consider. The components associated with the two different operators are depicted in two different colors: blue \ac{BS}, \ac{RIS}, \ac{UE} and the channels belong to operator 1, while the red ones belong to operator 2. The operators can potentially use site-sharing (as in the figure) to reduce deployment costs but transmit over two disjoint narrow frequency bands to their respective serving \acp{BS}. Each \ac{RIS} has $N$ reflecting elements, and is dedicated to and controlled by a single operator but affects both bands. To focus on the fundamentals of pilot contamination, we consider an environment where the direct \ac{UE}-\ac{BS} paths are blocked, while the \ac{UE}-\ac{RIS} and \ac{RIS}-\ac{BS} paths are operational. Since the \acp{BS} and \acp{RIS} have fixed deployment locations, we assume the \ac{RIS}-\ac{BS} channels $\mx{h}_k \in \mathbb{C}^N$ are known, while the \ac{UE}-\ac{RIS} channels $\mx{g}_k \in \mathbb{C}^N$ are unknown and to be estimated, for $k = 1,2$.

\begin{figure}
    \centering
    \includegraphics[width=\linewidth]{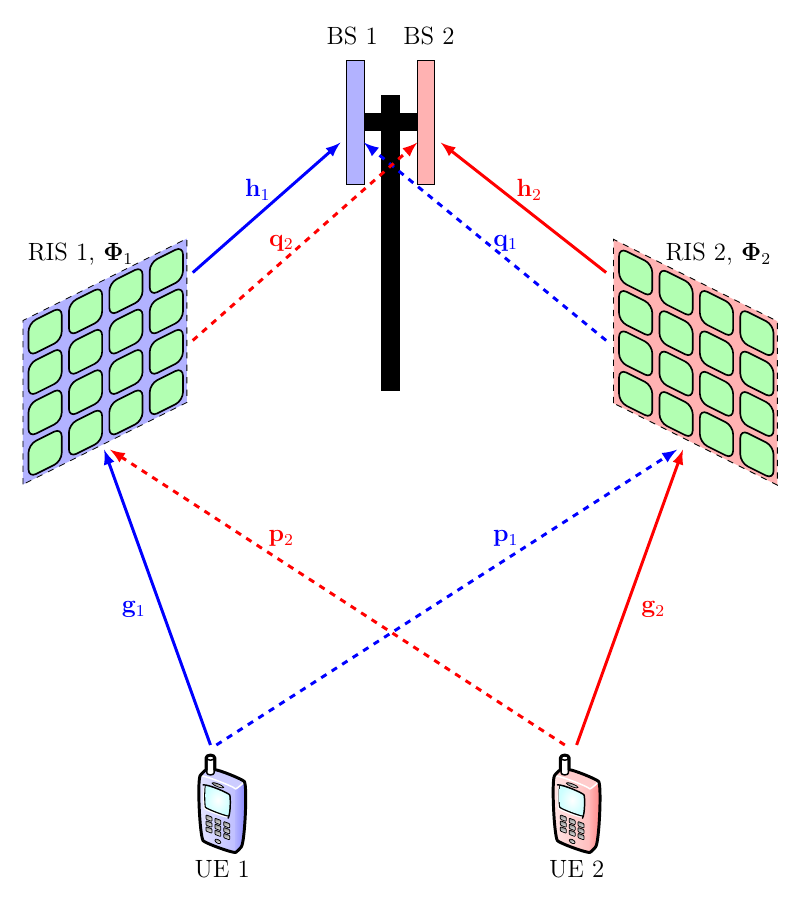}
    \caption{The considered setup with two \acp{UE}, two \acp{RIS}, and two co-located single-antenna \acp{BS}. The blue channels correspond to frequency band 1, and the red channels correspond to frequency band 2, subscribed by \ac{UE}s $1$ and $2$, respectively. The desired channels are denoted by solid lines, while the undesired channels (whose existences might be unknown to the \acp{BS}) are denoted by dashed lines. Each channel vector is $N$-dimensional because each \ac{RIS} has $N$ elements.}
    \label{fig:systemSetup}
\end{figure}

The signal transmitted by \ac{UE} $k$ reaches its serving \ac{BS} through  the channels $\mx{h}_k$ and $\mx{g}_k$, for $k = 1,2$. Importantly, each \ac{UE}'s transmitted signal is also reflected by the non-serving operator's \ac{RIS} and parts of the reflected signal will reach  the serving \ac{BS}. This effect contaminates the pilot signal reflected by the serving \ac{RIS} and we will study the implications. The phenomenon is illustrated in Fig.~\ref{fig:systemSetup}, where the resulting \ac{UE}-\ac{RIS} and \ac{RIS}-\ac{BS} channels are denoted by $\mx{p}_k  \in \mathbb{C}^N$ and $\mx{q}_k \in \mathbb{C}^N$, respectively, for $k = 1,2$.
Defining the pilot signal of \ac{UE} $k$ as $s_k\in\mathbb{C}$, the received signals on bands $1$ and $2$ at the \acp{BS} can be expressed as
\begin{subequations}\label{eq:rxPilot}
\begin{align}
    &y_{p1} = \sqrt{P_p}\mx{h}_1^T\boldsymbol{\Phi}_1\mx{g}_1s_1+\sqrt{P_p}\mx{q}_1^T\boldsymbol{\Phi}_2\mx{p}_1s_1+w_{p1},\\
    &y_{p2} = \sqrt{P_p}\mx{h}_2^T\boldsymbol{\Phi}_2\mx{g}_2s_2+\sqrt{P_p}\mx{q}_2^T\boldsymbol{\Phi}_1\mx{p}_2s_2+w_{p2},
\end{align}
\end{subequations}
where $y_{pk} \in \mathbb{C}$ denotes the received signal, $w_{pk}\sim\mathcal{CN}(0,1)$ denotes the receiver noise for band $k$, and $\boldsymbol{\Phi}_k = \text{diag}(e^{-j\phi_{k1}},\dots,e^{-j\phi_{kN}})$ denotes the $k$th \ac{RIS}'s reflection matrix, and $P_p$ denotes the pilot signal's transmission power. We assume $s_1=s_2=1$ without loss of generality. When analyzing channel estimation, it is more convenient to rewrite \eqref{eq:rxPilot} as
\begin{subequations}\label{eq:rxPilot2}
\begin{align}
    &y_{p1} = \sqrt{P_p}\boldsymbol{\phi}_1^T\mx{D}_{\mx{h}_1}\mx{g}_1+\sqrt{P_p}\boldsymbol{\phi}_2^T\mx{D}_{\mx{q}_1}\mx{p}_1+w_{p1},\\
    &y_{p2} = \sqrt{P_p}\boldsymbol{\phi}_2^T\mx{D}_{\mx{h}_2}\mx{g}_2+\sqrt{P_p}\boldsymbol{\phi}_1^T\mx{D}_{\mx{q}_2}\mx{p}_2+w_{p2},
\end{align}
\end{subequations}
where $\mx{D}_{\mx{h}_k}$ and $\mx{D}_{\mx{q}_k}$ represent the $N \times N$ diagonal matrices containing the elements of $\mx{h}_k$ and $\mx{q}_k$, and $\boldsymbol{\phi}_k \in \mathbb{C}^{N}$ denotes the column vectors containing the diagonal entries of $\boldsymbol{\Phi}_k$ for $ k = 1,2$.

As there are $N$ parameters in $\mx{g}_1$ and $\mx{g}_2$, at least $N$ linearly independent observations are needed to estimate them uniquely. To this end, we perform $L\geq N$ pilot transmissions over time, and we vertically stack the received signals to obtain
\begin{subequations}\label{eq:rxPilotVec}
\begin{align}
    &\mx{y}_{p1}=\sqrt{P_p}\mx{B}_1\mx{D}_{\mx{h}_1}\mx{g}_1+\sqrt{P_p}\mx{B}_2\mx{D}_{\mx{q}_1}\mx{p}_1+\mx{w}_{p1},\\
    &\mx{y}_{p2}=\sqrt{P_p}\mx{B}_2\mx{D}_{\mx{h}_2}\mx{g}_2+\sqrt{P_p}\mx{B}_1\mx{D}_{\mx{q}_2}\mx{p}_2+\mx{w}_{p2},
\end{align}
\end{subequations}
where $\mx{y}_{pk}=[y_{pk}[1],\dots,y_{pk}[L]]^T \in\mathbb{C}^{L}$ denotes the sequence of received signals from the $k$th \ac{UE} over $L$ time instances, and the matrices $\mx{B}_1$ and $\mx{B}_2$ represent the sequence of \ac{RIS} configurations over $L$ time instances; that is, $\mx{B}_k \triangleq \begin{bmatrix}\boldsymbol{\phi}_k[1] & \dots & \boldsymbol{\phi}_k[L]\end{bmatrix}^T \in\mathbb{C}^{L\times N}$
for $k = 1,2$.

In the remainder of this paper, we will analyze channel estimation and the resulting communication performance for deterministic and fading channels, respectively.

\section{Maximum Likelihood Estimation of Deterministic Channels}
\label{sec:MLE}

In this section, we will consider channel estimation for deterministic channels. The same assumptions and results will then be considered in
Section \ref{sec:DataTransmission} for data transmission.

We assume that $\mx{g}_k$ is a deterministic and unknown channel without any known structure. That is, $\mx{g}_k$ is an $N\times1$ vector of complex deterministic parameters to be estimated. In addition, we assume that $\mx{h}_k$ is perfectly known. On the other hand, the \acp{BS} do not know the existence of $\mx{q}_k$ and $\mx{p}_k$. Consequently, we can denote the received signal models assumed by the \acp{BS} as
\begin{subequations}\label{eq:rxPilotHat}
\begin{align}
    &\hat{\mathbf{y}}_{p1}=\sqrt{P_p}\mathbf{B}_1\mathbf{D}_{\mathbf{h}_1}\mathbf{g}_1+\mathbf{w}_{p1},\\
    &\hat{\mathbf{y}}_{p2}=\sqrt{P_p}\mathbf{B}_2\mathbf{D}_{\mathbf{h}_2}\mathbf{g}_2+\mathbf{w}_{p2}.
\end{align}
\end{subequations}
Since $\mx{g}_k$ does not have a known structure and hence consists of $N$ scalars, \ac{BS} $k$ requires at least $N$ independent observations to estimate it. To this end, both $\mathbf{B}_1,\mathbf{B}_2\in\mathbb{C}^{L\times N}$ must have full column rank. Furthermore, we require that the \ac{RIS} configurations used at different time instances are mutually orthogonal and contain entries on the unit circle that can be realized using a reflecting element. These assumptions result in $\mx{B}_k^H\mx{B}_k=L\mx{I}_{N}$. In classical non-Bayesian parameter estimation, the \ac{ML} estimator is widely used, which maximizes the likelihood function of the received observation over the unknown parameter. Since the \acp{BS} have misspecified received pilot signal models, they will instead maximize the likelihood functions obtained from the misspecified model, leading to \ac{MML} estimators. For \eqref{eq:rxPilotHat}, the \ac{MML} estimator can be expressed as
\begin{align}\nonumber
    \hat{\mx{g}}_k &= \arg\max_{\mx{g}_k}f(\mx{y}_{pk};\mx{g}_k)\\
    &=\arg\max_{\mx{g}_k}\frac{1}{(\pi\sigma_w^2)^{L}}\exp{\left(-\frac{\|\mx{y}_{pk}-\sqrt{P_p}\mx{B}_k\mx{D}_{\mx{h}_k}\mx{g}_k\|^2}{\sigma_w^2}\right)}\nonumber\\
    &=\arg\min_{\mx{g}_k}\left\|\mx{y}_{pk}-\sqrt{P_p}\mx{B}_k\mx{D}_{\mx{h}_k}\mx{g}_k\right\|^2\nonumber\\
    &=\frac{1}{\sqrt{P_p}}\mathbf{D}_{\mathbf{h}_k}^{-1}(\mathbf{B}_k^H\mathbf{B}_k)^{-1}\mathbf{B}_k^H\mathbf{y}_{pk}\nonumber\\
    &=\frac{1}{L\sqrt{P_p}}\mathbf{D}_{\mathbf{h}_k}^{-1}\mathbf{B}_k^H\mathbf{y}_{pk}.\label{eq:gHatMLmismatch}
\end{align}
In the following subsections, we describe the behavior of this estimator for two different choices of the $\mx{B}_k$ matrices.

\subsection{Case 1: The \acp{RIS} Adopt the Same 
Configuration Sequence}
We discussed in the introduction that in the absence of inter-operator cooperation, it is highly likely that the \acp{RIS} will use the same standardized sequence of configurations during the channel estimation phase, which corresponds to $\mx{B}_1=\mx{B}_2=\mx{B}$.\footnote{The analysis in this paper can be easily extended to the case when $\mx{B}_1= \mx{U}\mx{B}_2$ for some unitary matrix $\mx{U}$, so that the configuration sequences have identical spans. It is the overlap of the spans that can cause issues.} In this case, \eqref{eq:gHatMLmismatch} becomes
\begin{align}
    \hat{\mx{g}}_k = \mx{g}_k + \mathbf{D}_{\mathbf{h}_k}^{-1}\mathbf{D}_{\mathbf{q}_k}\mathbf{p}_k+\frac{1}{L\sqrt{P}_p}\mathbf{D}_{\mx{h}_k}^{-1}\mathbf{B}^H\mathbf{w}_{pk}.
\end{align}
Since we consider the channels as deterministic parameters, we obtain the probability distribution 
\begin{equation}\label{eq:gHat_1}
    \hat{\mx{g}}_k\sim\mathcal{CN}\!\left(\mx{g}_k + \mathbf{D}_{\mathbf{h}_k}^{-1}\mathbf{D}_{\mathbf{q}_k}\mathbf{p}_k,\frac{\sigma_w^2}{LP_p}(\mathbf{D}_{\mathbf{h}_k}^H\mathbf{D}_{\mathbf{h}_k})^{-1}\!\right).
\end{equation}
We notice that $\hat{\mx{g}}_k$ is biased; that is, $\mx{b}_k\triangleq\mathbb{E}[\hat{\mx{g}}_k-\mx{g}_k]=\mathbf{D}_{\mathbf{h}_k}^{-1}\mathbf{D}_{\mathbf{q}_k}\mathbf{p}_k\neq\boldsymbol{0}$. The estimator bias does not vanish when increasing $P_p$ or $L$,  or decreasing $\sigma_w^2$. Hence, this estimator is not asymptotically unbiased. This is a new instance of an extensively studied phenomenon in the massive \ac{MIMO} literature, namely pilot contamination \cite{Marzetta2010a,Sanguinetti2020a}.
Interestingly, the \acp{RIS} cause pilot contamination even between two non-overlapping frequency bands, which has not been recognized in the previous literature.

\subsection{Case 2: 
The \acp{RIS} Adopt Different Configuration Sequences}
In this section, we consider the generic case of $\mx{B}_1\neq\mx{B}_2$. To motivate the proposed method for configuring $\mx{B}_1$ and $\mx{B}_2$, we first consider the case where the \acp{BS} are aware of the true signal model in \eqref{eq:rxPilotVec}, and therefore can estimate both $\mx{g}_k$ and $\mx{r}_k\triangleq\mx{D}_{\mx{q}_k}\mx{p_k}$. The resulting system model can be expressed as
\begin{subequations}\label{eq:fullCaseChanEst}
    \begin{align}
        &\mx{y}_{p1} = \sqrt{P}_p\begin{bmatrix}
            \mx{B}_1\mx{D}_{\mx{h}_1} & \mx{B}_2
        \end{bmatrix}\begin{bmatrix}
            \mx{g}_1\\\mx{r}_1
        \end{bmatrix} + \mx{w}_{p1},\\
        &\mx{y}_{p2} = \sqrt{P}_p\begin{bmatrix}
            \mx{B}_2\mx{D}_{\mx{h}_2} & \mx{B}_1
        \end{bmatrix}\begin{bmatrix}
            \mx{g}_2\\\mx{r}_2
        \end{bmatrix} + \mx{w}_{p2}.
    \end{align}    
\end{subequations}
In \eqref{eq:fullCaseChanEst}, a known linear transformation is applied to the parameter vector of interest in the presence of additive noise. Consequently, the \ac{ML} estimates of \ac{UE} 1's channels become
\begin{equation}\label{eq:fullCaseML}
    \begin{bmatrix}
        \hat{\mx{g}}_1\\\hat{\mx{r}}_1
    \end{bmatrix} = \frac{1}{\sqrt{P_p}}\begin{bmatrix}
        L\mx{D}_{\mx{h}_1}^H\mx{D}_{\mx{h}_1} & \mx{D}_{\mx{h}_1}^H\mx{B}_1^H\mx{B}_2\\
        \mx{B}_2^H\mx{B}_1\mx{D}_{\mx{h}_1} & L\mx{I}_{N}
    \end{bmatrix}^{-1}\begin{bmatrix}
        \mx{D}_{\mx{h}_1}^H\mx{B}_1^H \\ \mx{B}_2^H
    \end{bmatrix}\mx{y}_{p1}.
\end{equation}
Note that in this case, the total dimension of the unknown parameter vector is $2N$, hence, at least $2N$ independent observations are required for the matrix inverse to exist.\footnote{The pseudo-inverse could be used when there are fewer observations, but it will not provide a useful estimate.} The structure in \eqref{eq:fullCaseML} applies to \ac{UE} 2 with alternated indices, and it gives the \ac{ML} estimator, which is both unbiased and efficient, since \eqref{eq:fullCaseChanEst} is a linear observation model with additive Gaussian noise \cite[Th.~7.3]{kayestimation}. Hence, \eqref{eq:fullCaseML} is unbiased irrespective of other parameters such as $\sigma_w^2$, $L$, and $P_p$, and it achieves the \ac{CRLB}, which provides a lower bound on the \ac{MSE} of any unbiased estimator \cite{poorbook}. It has to be noted  that when $\mx{B}_1^H\mx{B}_2 = \boldsymbol{0}$, \eqref{eq:fullCaseML} becomes
\begin{align}\nonumber\label{eq:fullCaseMLOrth}
    \begin{bmatrix}
        \hat{\mx{g}}_1\\\hat{\mx{r}}_1
    \end{bmatrix} &= \frac{1}{L\sqrt{P_p}}\begin{bmatrix}
        \mx{D}_{\mx{h}_1}^H\mx{D}_{\mx{h}_1} & \boldsymbol{0}\\
        \boldsymbol{0} & \mx{I}_{N}
    \end{bmatrix}^{-1}\begin{bmatrix}
        \mx{D}_{\mx{h}_1}^H\mx{B}_1^H \\ \mx{B}_2^H
    \end{bmatrix}\mx{y}_{p1}\\
    &= \frac{1}{L\sqrt{P_p}}\begin{bmatrix}
        \mx{D}_{\mx{h}_1}^{-1}\mx{B}_1^H \\ \mx{B}_2^H
    \end{bmatrix}\mx{y}_{p1}.
\end{align}
Note that the expression for $\hat{\mx{g}}_1$ in \eqref{eq:fullCaseMLOrth} is the same as in \eqref{eq:gHatMLmismatch}. This shows that when $\mx{B}_1^H\mx{B}_2 = \boldsymbol{0}$, the \ac{MML} in \eqref{eq:gHatMLmismatch} coincides with the \ac{ML} estimator; that is, the misspecified model is sufficient when the configuration sequences are designed to alleviate pilot interference because the missing terms anyway vanish in the receiver processing. The probability distribution of $\hat{\mx{g}}_k$ in this case can be expressed as
\begin{equation}\label{eq:gHat_2}
    \hat{\mx{g}}_k\sim\mathcal{CN}\left(\mx{g}_k,\frac{\sigma_w^2}{LP_p}(\mathbf{D}_{\mathbf{h}_k}^H\mathbf{D}_{\mathbf{h}_k})^{-1}\right),
\end{equation}
which shows that choosing the \ac{RIS} configuration sequences such that $\mx{B}_1$ and $\mx{B}_2$ remove the bias from the \ac{MML} estimator. However, the major setback of this approach is that the minimum number of observations required for this  channel estimation procedure is $2N$ instead of $N$, due to the fact that the $2N$-many $L$-dimensional columns must all be mutually orthogonal, for which $L\geq2N$ is required. Considering that the estimator bias in \eqref{eq:gHat_1} does not vanish with increasing $L$, this is a necessary sacrifice. Hence, it has to be noted that the number of pilot transmissions increases linearly with the number of \acp{RIS} deployed in proximity. 
\subsection{\ac{MSE} During Channel Estimation}
The estimation error can be quantified through the \ac{MSE}, which is the trace of the error covariance matrix. We will derive the \ac{MSE} in this section. In this derivation, we do not assume a particular choice of $\mx{B}_1$ and $\mx{B}_2$, but we utilize the basic assumption $\mx{B}_k\mx{B}_k^H=L\mx{I}_N$. Consequently, we use $\mx{b}_k$ to denote the potential  estimator bias. We can then compute the error covariance matrix as
\begin{align}\nonumber
    \boldsymbol{\Sigma}_{e,k} &= \mathbb{E}\left[(\hat{\mx{g}}_k-\mx{g}_k)(\hat{\mx{g}}_k-\mx{g}_k)^H\right]\\
    &=\mx{b}_k\mx{b}_k^H+\frac{1}{LP_p}\mathbb{E}\left[\mx{D}_{\mx{h}_k}^{-1}\mx{w}_{pk}\mx{w}_{pk}^H\mx{D}_{\mx{h}_k}^{-H}\right]\nonumber\\
    &=\mx{b}_k\mx{b}_k^H+\frac{\sigma_w^2}{LP_p}\left(\mx{D}_{\mx{h}_k}\mx{D}_{\mx{h}_k}^{H} \right)^{-1}.
\end{align}
Consequently, the trace of the error covariance matrix becomes
\begin{equation}\label{eq:traceError}
    \mathrm{tr}(\boldsymbol{\Sigma}_{e,k}) = \|\mx{b}_k\|^2+\frac{\sigma_w^2}{LP_p}\sum_{n=1}^N\frac{1}{|h_{kn}|^2}.
\end{equation}
Note that for high $P_p$, $L$, and/or low $\sigma_w^2$, the second term in \eqref{eq:traceError} vanishes, and the trace of the error covariance converges to $\|\mx{b}_k\|^2$ which depends on the configurations of $\mx{B}_1,\mx{B}_2$. For the two previously considered cases, we have
\begin{equation}\label{eq:chEstFloor}
    \|\mx{b}_k\|^2 = \begin{cases}
        \sum_{n=1}^{N}\frac{|r_{kn}|^2}{|h_{kn}|^2} & \mx{B}_1 = \mx{B}_2,\\
        0 & \mx{B}_1^H\mx{B}_2 = \boldsymbol{0}.
    \end{cases}
\end{equation}
This result shows that configuring the \ac{RIS}s such that $\mx{B}_1^H\mx{B}_2=\boldsymbol{0}$ removes the asymptotic floor on the average \ac{MSE}, which comes from the energy of the estimator bias. On the other hand, when the intended \ac{RIS}-\ac{BS} links $\mx{h}_1,\mx{h}_2$ are strong relative to the unintended and unknown overall link $\mx{r}_k$, the estimator bias will be weaker and the performance loss associated with choosing $\mx{B}_1=\mx{B}_2$ will be lower. Nevertheless, pilot contamination results in a fundamental error floor, even if the \acp{RIS} are utilized in different bands. In the next section, we consider the estimation of data based on the channel estimation performed in this section and analyze the consequence of pilot contamination in this phase.

\section{Data Signal Estimation with Deterministic Channels}\label{sec:DataTransmission}

The channel estimation is followed by data transmission over the same deterministic channel as in Section~\ref{sec:MLE}. The receiver can use the channel estimate derived in the last section when determining the transmitted signal. Practical channels are never fully deterministic but might have a long coherence time. Moreover, the impact of estimation errors is only relevant when the data packet has a modest size so we cannot afford to spend much resources on pilots. For this reason, we cannot consider the channel capacity as performance metric but will instead consider the \ac{MSE}.

Defining the data signal transmitted by the $k$th \ac{UE} as $x_k\sim\mathcal{CN}(0,1)$, we can express the received data as
\begin{subequations}\label{eq:rxData}
    \begin{align}
        &y_1 = \sqrt{P_d}(\mx{h}_1^T\hat{\boldsymbol{\Phi}}_1\mx{g}_1+\mx{q}_1^T\hat{\boldsymbol{\Phi}}_2\mx{p}_1)x_1 + w_1,\\
        &y_2 = \sqrt{P_d}(\mx{h}_2^T\hat{\boldsymbol{\Phi}}_2\mx{g}_2+\mx{q}_2^T\hat{\boldsymbol{\Phi}}_1\mx{p}_2)x_2 + w_2,
    \end{align}
\end{subequations}
where $w_k\sim\mathcal{CN}(0,\sigma_w^2)$ denotes the receiver noise, $P_d$ denotes the data transmission power, and the RIS configuration matrices $\hat{\boldsymbol{\Phi}}_k$ are selected based on the estimated channels to maximize the average gain of the desired cascaded channel as shown in \cite[Sec.~II]{Bjornson2022b}:
\begin{align}\nonumber
    &\hat{\phi}_{kn} = \arg(h_{kn}) + \arg(\hat{g}_{kn}),\\
    &\hat{\boldsymbol{\Phi}}_k = \mathrm{diag}\left(e^{-j\hat{\phi}_{k1}},\dots,e^{-j\hat{\phi}_{kN}}\right).\label{eq:RISphasematch}
\end{align}
However, since the \acp{BS} are unaware of the unintended reflections and base their data reception on the previously obtained channel estimates, they assume the following misspecified received data signal models:
\begin{subequations}\label{eq:rxDataHat}
    \begin{align}
        &\hat{y}_1 = \sqrt{P_d}\mx{h}_1^T\hat{\boldsymbol{\Phi}}_1\hat{\mx{g}}_1x_1 + w_1,\\
        &\hat{y}_2 = \sqrt{P_d}\mx{h}_2^T\hat{\boldsymbol{\Phi}}_2\hat{\mx{g}}_2x_2 + w_2.
    \end{align}
\end{subequations}
Introducing the notation $m_k\triangleq\sqrt{P_d}(\mx{h}_k^T\hat{\boldsymbol{\Phi}}_k\mx{g}_k+\mx{q}_k^T\hat{\boldsymbol{\Phi}}_j\mx{p}_k)$ for $j,k\in\{1,2\}, j\neq k$, and $\hat{m}_k\triangleq\sqrt{P_d}\mx{h}_k^T\hat{\boldsymbol{\Phi}}_k\hat{\mx{g}}_k$, \eqref{eq:rxData} and \eqref{eq:rxDataHat} can be expressed as
\begin{subequations}
    \begin{align}
        &y_k = m_kx_k+w_k,\quad k = 1,2,\label{eq:mData}\\
        &\hat{y}_k = \hat{m}_kx_k+w_k,\quad k = 1,2.\label{eq:mDataHat}
    \end{align}
\end{subequations}
Based on the misspecified observation model in \eqref{eq:mDataHat}, the \acp{BS} estimate $x_k$ by using the misspecified \ac{MMSE} estimator
\begin{equation}
    \hat{x}_k = \frac{\hat{m}_k^*}{|\hat{m}_k|^2+\sigma_w^2}y_k,\quad k = 1,2.
\end{equation}
In this section, we consider the \ac{MSE} between $x_k$ and $\hat{x}_k$ as the performance metric for the data transmission. We derive the data estimation \ac{MSE} for \ac{UE} $k$ as
\begin{align}\nonumber
    &\mathbb{E}\left[|x_k-\hat{x}_k|^2\right] = 1 + \mathbb{E}\left[|\hat{x}_k|^2\right]-2\mathrm{Re}(\mathbb{E}[x_k\hat{x}_k^*])\\
    &=1+\mathbb{E}\left[\frac{|\hat{m}_k|^2(|m_k|^2+\sigma_w^2)}{(|\hat{m}_k|^2+\sigma_w^2)^2}\right]\nonumber\\
    &-2\mathrm{Re}\left(\mathbb{E}\left[\frac{\hat{m}_km_k^*}{|\hat{m}_k|^2+\sigma_w^2}\right]\right)\nonumber\\
    &=1+\mathbb{E}\left[\frac{|\hat{m}_k|^2(|m_k|^2+\sigma_w^2)-2\mathrm{Re}(\hat{m}_1m_1^*)(|\hat{m}_k|^2+\sigma_w^2)}{(|\hat{m}_k|^2+\sigma_w^2)^2}\right]\nonumber\\
    &=1+\mathbb{E}\left[\frac{|\hat{m}_k|^2(|m_k|^2+\sigma_w^2)-2\mathrm{Re}(\hat{m}_1m_1^*)(|\hat{m}_k|^2+\sigma_w^2)}{(|\hat{m}_k|^2+\sigma_w^2)^2}\right]\nonumber\\
    &+\mathbb{E}\left[\frac{\sigma_w^2(|m_k|^2+\sigma_w^2)-\sigma_w^2(|m_k|^2+\sigma_w^2)}{(|\hat{m}_k|^2+\sigma_w^2)^2}\right]\nonumber\\
    &=1+\mathbb{E}\left[\frac{(|\hat{m}_k|^2+\sigma_w^2)(|m_k|^2+\sigma_w^2-2\mathrm{Re}(\hat{m}_km_k^*))}{(\hat{m}_k^2+\sigma_w^2)^2}\right]\nonumber\\
    &-\mathbb{E}\left[\frac{\sigma_w^2(|m_k|^2+\sigma_w^2)}{(\hat{m}_k^2+\sigma_w^2)^2}\right]\nonumber\\
    &=\mathbb{E}\left[\frac{|m_k-\hat{m}_k|^2+2\sigma_w^2}{|\hat{m}_k|^2+\sigma_w^2}\right]-\mathbb{E}\left[\frac{\sigma_w^2(|m_k|^2+\sigma_w^2)}{(|\hat{m}_k|^2+\sigma_w^2)^2}\right].\label{eq:dataMSE}
\end{align}
Defining $\epsilon_k\triangleq m_k-\hat{m}_k$, \eqref{eq:dataMSE} can be rewritten as
\begin{equation}\label{eq:dataMSE_v2}
    \mathbb{E}[|x_k-\hat{x}_k|^2] = \mathbb{E}\left[\frac{|\epsilon_k|^2+2\sigma_w^2}{|m_k-\epsilon_k|^2+\sigma_w^2}-\frac{\sigma_w^2(|m_k|^2+\sigma_w^2)}{(|m_k-\epsilon_k|^2+\sigma_w^2)^2}\right].
\end{equation}
To examine the impact of pilot contamination on the data estimation performance more clearly, we now consider  channel estimation at high \acp{SNR}, so that the estimation error only comes from the estimator bias, i.e., pilot contamination. This happens  when $L$ or $P_p$ is high and/or $\sigma_w^2$ is low, which results in the estimator covariances in \eqref{eq:gHat_1} and \eqref{eq:gHat_2} becoming zero.  For notational convenience, we consider the case where $P_p$ is arbitrarily large so that  $\lim_{P_p\to\infty}\hat{\mx{g}}_k = \mx{g}_k+\mx{b}_k$, where
\begin{equation}\label{eq:biasCases}
    \mx{b}_k = \begin{cases}
        \mx{D}_{\mx{h}_k}^{-1}\mx{D}_{\mx{q}_k}\mx{p}_k & \mx{B}_1=\mx{B}_2,\\
        \boldsymbol{0} & \mx{B}_1^H\mx{B}_2=\boldsymbol{0}.
    \end{cases}
\end{equation}

\subsection{Data \ac{MSE} with Channel Estimation at High \ac{SNR}}\label{sec:highSNRChanEstRef}
In \eqref{eq:dataMSE_v2}, $\epsilon_k$ and $m_k$ are functions of $\hat{\mx{g}}_1$ and $\hat{\mx{g}}_2$, therefore as $\hat{\mx{g}}_1$ and $\hat{\mx{g}}_2$ converge to their means, $\epsilon_k$ and $m_k$ become
\begin{subequations}\label{eq:barchannels}
\begin{align}
    &\overline{m}_k = \sqrt{P_d}(\mx{h}_k^T\Bar{\boldsymbol{\Phi}}_k\mx{g}_k+\mx{q}_k^T\Bar{\boldsymbol{\Phi}}_j\mx{p}_k),\\
    &\overline{\epsilon}_k = \sqrt{P_d}(\mx{q}_k^T\Bar{\boldsymbol{\Phi}}_j\mx{p}_k-\mx{h}_k^T\Bar{\boldsymbol{\Phi}}_k\mx{b}_k),
\end{align}    
\end{subequations}
for $j,k\in\{1,2\}$ and $j\neq k$. $\Bar{\boldsymbol{\Phi}}_k$ denotes the \ac{RIS} configuration computed according to \eqref{eq:RISphasematch} with $\hat{\mx{g}}_k=\mx{g}_k+\mx{b}_k$. At high \ac{SNR}, the MSE in \eqref{eq:dataMSE_v2} can be rewritten as
\begin{equation}\label{eq:dataMSE_highChan}
    \text{MSE} = \frac{|\overline{\epsilon}_k|^2+2\sigma_w^2}{|\overline{m}_k-\overline{\epsilon}_k|^2+\sigma_w^2}-\frac{\sigma_w^2(|\overline{m}_k|^2+\sigma_w^2)}{(|\overline{m}_k-\overline{\epsilon}_k|^2+\sigma_w^2)^2}.
\end{equation}
This is a practically achievable limit since RIS-aided systems require long pilot sequences over a narrow bandwidth, thus, the effective SNR (proportional to $P_p L$) during pilot transmission can be much larger than in the data transmission phase.

\subsection{Data \ac{MSE} with Transmission at 
High \ac{SNR}}
In the previous section, we obtained the expression for the data \ac{MSE} when  the channels are estimated with a high pilot \ac{SNR}, while the data transmission is done at an arbitrary SNR. To study the case when also the data transmission is conducted at a high SNR, we let $\sigma_w^2 \to 0$, which results in the limit
\begin{equation}\label{eq:highSNRfloorData}
    \lim_{\sigma_w^2\to 0}\text{MSE} = \frac{|\overline{\epsilon}_k|^2}{|\overline{m}_k-\overline{\epsilon}_k|^2}.
\end{equation}
Note that the resulting expression denotes the ratio between the estimated overall \ac{SISO} channel $\hat{m}_k$'s power and the mismatch parameter $\epsilon_k$'s power. Recall that $\mx{b}_k$ depends on which sequence of \ac{RIS} configurations is utilized. For $\mx{B}_1=\mx{B}_2$, we can obtain $\epsilon_k$ as
\begin{align}\nonumber
    \epsilon_k &= \sqrt{P_d}\mx{q}_k^T\hat{\boldsymbol{\Phi}}_j\mx{p}_k-\sqrt{P_d}\mx{h}_k^T\hat{\boldsymbol{\Phi}}_k\mx{D}_{\mx{h}_1}^{-1}\mx{D}_{\mx{q}_1}\mx{p}_1\nonumber\\
    &=\sqrt{P_d}\mx{q}_k^T\hat{\boldsymbol{\Phi}}_j\mx{p}_k-\sqrt{P_d}\hat{\boldsymbol{\phi}}_k\mx{D}_{\mx{h}_k}\mx{D}_{\mx{h}_k}^{-1}\mx{D}_{\mx{q}_1}\mx{p}_1\nonumber\\
    &=\sqrt{P_d}\mx{q}_k^T\hat{\boldsymbol{\Phi}}_j\mx{p}_k-\sqrt{P_d}\hat{\boldsymbol{\phi}}_k\mx{D}_{\mx{q}_1}\mx{p}_1\nonumber\\
    &=\sqrt{P_d}\mx{q}_k^T\hat{\boldsymbol{\Phi}}_j\mx{p}_k-\sqrt{P_d}\mx{q}_k^T\hat{\boldsymbol{\Phi}}_k\mx{p}_k\nonumber\\
    &=\sqrt{P_d}\mx{q}_k^T(\hat{\boldsymbol{\Phi}}_j-\hat{\boldsymbol{\Phi}}_k)\mx{p}_k.
\end{align}
On the other hand, $\mx{B}_1^H\mx{B}_2=\boldsymbol{0}$ removes $\mx{b}_k$ for $k\in\{1,2\}$. Consequently, we have
\begin{equation}\label{eq:epsilonExplicit}
    \epsilon_k = \begin{cases}
        \sqrt{P_d}\mx{q}_k^T(\hat{\boldsymbol{\Phi}}_j-\hat{\boldsymbol{\Phi}}_k)\mx{p}_k & \mx{B}_1=\mx{B}_2,\\
        \sqrt{P_d}\mx{q}_k^T\hat{\boldsymbol{\Phi}}_j\mx{p}_k & \mx{B}_1^H\mx{B}_2=\boldsymbol{0}.
    \end{cases}
\end{equation}
We notice that $\epsilon_k$ corresponds to only the unintended reflection path when $\mx{B}_1^H\mx{B}_2=\boldsymbol{0}$. On the other hand, $\mx{B}_1=\mx{B}_2$ yields an expression depending on the difference between the two \acp{RIS}' configurations during data transmission. Since each \ac{RIS} is configured based on the (estimated) channels of their respective users, it is highly unlikely that the configurations will be close. Moreover, it has to be noted that the \ac{RIS} configuration of the non-serving \ac{RIS} is different between the two cases since the channel estimates are also different.

\section{Channel Estimation based on Correlated Rayleigh Fading Priors}
\label{sec:RayleighChan}

We now switch focus to consider fading channels that can be modeled using the Bayesian framework.
In this section, we consider channel estimation and assume that all the \ac{UE}-\ac{RIS} channels exhibit spatially correlated Rayleigh fading: $\mx{g}_k\sim\mathcal{CN}(\boldsymbol{0},\boldsymbol{\Sigma}_{\mx{g}_k})$ and $\mx{p}_k\sim\mathcal{CN}(\boldsymbol{0},\boldsymbol{\Sigma}_{\mx{p}_k})$. 
The covariance matrices $\boldsymbol{\Sigma}_{\mx{g}_k},\boldsymbol{\Sigma}_{\mx{p}_k} \in \mathbb{C}^{N \times N}$ are generic positive semi-definite matrices.
In addition, the \acp{BS} know $\mx{h}_k$ perfectly while they consider $\mx{q}_k$ as deterministic and unknown channels for $k = 1,2$. The pilot transmission model assumed by the \acp{BS} for $k = 1,2$ can be expressed as
\begin{subequations}\label{eq:misPilotBayes}
    \begin{align}
        &\hat{\mx{y}}_{p1}=\sqrt{P_p}\mx{B}_1\mx{D}_{\mx{h}_1}\mx{g}_{1}+\mx{w}_{p1}\in\mathbb{C}^L,\\
        &\hat{\mx{y}}_{p2}=\sqrt{P_p}\mx{B}_2\mx{D}_{\mx{h}_2}\mx{g}_{2}+\mx{w}_{p2}\in\mathbb{C}^L.
    \end{align}
\end{subequations}
Based on \eqref{eq:misPilotBayes}, the \acp{BS} can estimate $\mx{g}_k$ via a misspecified \ac{MMSE} estimator, which can be expressed as
\begin{subequations}
    \begin{align}
        \hat{\mathbf{g}}_1 &= \frac{1}{\sqrt{P_p}}\boldsymbol{\Sigma}_{\mx{g}_1}\mx{D}_{\mx{h}_1}^H\mathbf{B}_1^H\nonumber\\
        &\times
        \left(\mathbf{B}_1\mx{D}_{\mx{h}_1}\boldsymbol{\Sigma}_{\mx{g}_1}\mx{D}_{\mx{h}_1}^H\mathbf{B}_1^H+\frac{\sigma_w^2}{P_p}\mathbf{I}_L\right)^{-1}\mathbf{y}_{p1},\\
        \hat{\mathbf{g}}_2 &= \frac{1}{\sqrt{P_p}}\boldsymbol{\Sigma}_{\mx{g}_2}\mx{D}_{\mx{h}_2}^H\mathbf{B}_2^H\nonumber\\
        &\times\left(\mathbf{B}_2\mx{D}_{\mx{h}_2}\boldsymbol{\Sigma}_{\mx{g}_2}\mx{D}_{\mx{h}_2}^H\mathbf{B}_2^H+\frac{\sigma_w^2}{P_p}\mathbf{I}_L\right)^{-1}\mathbf{y}_{p2}.
    \end{align}
\end{subequations}
The diagonal entries of the error covariance matrix represent the \acp{MSE} of the corresponding channel entry. We first define $\mx{r}_k\triangleq\mx{D}_{\mx{q}_k}\mx{p}_k$, which results in $\mx{r}_k\sim\mathcal{CN}(\boldsymbol{0},\boldsymbol{\Sigma}_{\mx{r}_k})$ where $\boldsymbol{\Sigma}_{\mx{r}_k}\triangleq\mx{D}_{\mx{q}_k}\boldsymbol{\Sigma}_{\mx{p}_k}\mx{D}_{\mx{q}_k}^H$. To simplify the representation of the channel estimation error covariance matrix, we introduce the following notation:
\begin{subequations}
    \begin{align}
        &\mx{C}_{\mx{g}\hat{\mx{y}}}\triangleq \mathbb{E}[\mx{g}_k\hat{\mx{y}}_{pk}^H] =\mathbb{E}[\mx{g}_k\mx{y}_{pk}^H] =\sqrt{P_p}\boldsymbol{\Sigma}_{\mx{g}_k}\mx{D}_{\mx{h}_k}^H\mx{B}_k^H,\\
        &\mx{C}_{\hat{\mx{y}}\hat{\mx{y}}}\triangleq \mathbb{E}[\hat{\mx{y}}_{pk}\hat{\mx{y}}_{pk}^H] =P_p\mathbf{B}_k\mx{D}_{\mx{h}_k}\boldsymbol{\Sigma}_{\mathbf{g}_k}\mx{D}_{\mx{h}_k}^H\mathbf{B}_k^H+\sigma_w^2\mx{I}_L,\\
        &\mx{C}_{\mx{y}\mx{y}}\triangleq \mathbb{E}[\mx{y}_{pk}\mx{y}_{pk}^H] =\mx{C}_{\hat{\mx{y}}\hat{\mx{y}}} + P_p\mx{B}_j\boldsymbol{\Sigma}_{\mx{r}_k}\mx{B}_j^H.
    \end{align}
\end{subequations}
With this notation, the error covariance matrix can be represented as
\begin{align}\nonumber
    &\mathbb{E}[(\mathbf{g}_k-\hat{\mathbf{g}}_k)(\mathbf{g}_k-\hat{\mathbf{g}}_k)^H] = \\    &=\boldsymbol{\Sigma}_{\mathbf{g}_k}+\mx{C}_{\mx{g}\hat{\mx{y}}}\mx{C}_{\hat{\mx{y}}\hat{\mx{y}}}^{-1}\mx{C}_{\mx{y}\mx{y}}\mx{C}_{\hat{\mx{y}}\hat{\mx{y}}}^{-1}\mx{C}_{\mx{g}\hat{\mx{y}}}^H-2\mx{C}_{\mx{g}\hat{\mx{y}}}\mx{C}_{\hat{\mx{y}}\hat{\mx{y}}}^{-1}\mx{C}_{\mx{g}\hat{\mx{y}}}^H\nonumber\\
    &=\boldsymbol{\Sigma}_{\mathbf{g}_k}+\mx{C}_{\mx{g}\hat{\mx{y}}}\mx{C}_{\hat{\mx{y}}\hat{\mx{y}}}^{-1}(\mx{C}_{\hat{\mx{y}}\hat{\mx{y}}}+P_p\mx{B}_j\boldsymbol{\Sigma}_{\mx{r}_k}\mx{B}_j^H)\mx{C}_{\hat{\mx{y}}\hat{\mx{y}}}^{-1}\mx{C}_{\mx{g}\hat{\mx{y}}}^H\nonumber\\
    &-2\mx{C}_{\mx{g}\hat{\mx{y}}}\mx{C}_{\hat{\mx{y}}\hat{\mx{y}}}^{-1}\mx{C}_{\mx{g}\hat{\mx{y}}}^H\nonumber\\
    &=\underbrace{\boldsymbol{\Sigma}_{\mathbf{g}_k}-\mx{C}_{\mx{g}\hat{\mx{y}}}\mx{C}_{\hat{\mx{y}}\hat{\mx{y}}}^{-1}\mx{C}_{\mx{g}\hat{\mx{y}}}^H}_{\text{error covariance for } \mx{y}=\hat{\mx{y}}}+\underbrace{P_p\mx{C}_{\mx{g}\hat{\mx{y}}}\mx{C}_{\hat{\mx{y}}\hat{\mx{y}}}^{-1}\mx{B}_j\boldsymbol{\Sigma}_{\mx{r}_k}\mx{B}_j^H\mx{C}_{\hat{\mx{y}}\hat{\mx{y}}}^{-1}\mx{C}_{\mx{g}\hat{\mx{y}}}^H}_{\text{term coming from pilot contamination}}.\label{eq:chanEstMSE}
\end{align}
Note that the additive term in \eqref{eq:chanEstMSE} coming from pilot contamination depends on $\boldsymbol{\Sigma}_{\mx{r}_k}$. If we consider the case where the channel $\mx{r}_k$ does not exist, that is, $\boldsymbol{\Sigma}_{\mx{r}_k}=\boldsymbol{0}$, then we would obtain the error covariance in the form of a typical \ac{MMSE} estimation error covariance matrix. It is also important to address the dependence of the term coming from pilot contamination on $P_p$: while the terms $\mx{C}_{\mx{g}\hat{\mx{y}}}$ scale with $\sqrt{P_p}$ each, the terms $\mx{C}_{\hat{\mx{y}}\hat{\mx{y}}}^{-1}$ scale with $1/P_p$ each, and along with the leading $P_p$ multiplier, we can see that $P_p$-dependent terms cancel each other out, hence leaving a non-vanishing pilot contamination term.

\subsection{High-\ac{SNR} Channel Estimation}
Now, we investigate the behavior of the error covariance matrix in \eqref{eq:chanEstMSE} when $\sigma_w^2$ is low or $P_p$ is high. Note that the first term in \eqref{eq:chanEstMSE} is already the error covariance matrix for an \ac{MMSE} estimator without misspecification that estimates $\mx{g}_k$ based on the observation $\hat{\mx{y}}_{pk}$. Therefore, we know that this term vanishes at high \ac{SNR}. Consequently, the asymptotic error covariance matrix is governed by the high-\ac{SNR} behavior of the term coming from pilot contamination, that is
\begin{align}\nonumber
    &\lim_{\sigma_w^2\to 0}\mathbb{E}[(\mathbf{g}_k-\hat{\mathbf{g}}_k)(\mathbf{g}_k-\hat{\mathbf{g}}_k)^H]\\
    &=\lim_{\sigma_w^2\to 0}P_p\mx{C}_{\mx{g}\hat{\mx{y}}}\mx{C}_{\hat{\mx{y}}\hat{\mx{y}}}^{-1}\mx{B}_j\boldsymbol{\Sigma}_{\mx{r}_k}\mx{B}_j^H\mx{C}_{\hat{\mx{y}}\hat{\mx{y}}}^{-1}\mx{C}_{\mx{g}\hat{\mx{y}}}^H\nonumber\\
    &=P_p^2\boldsymbol{\Sigma}_{\mx{g}_k}\mx{D}_{\mx{h}_k}^H\mx{B}_k^H\left(P_p\mathbf{B}_k\mx{D}_{\mx{h}_k}\boldsymbol{\Sigma}_{\mathbf{g}_k}\mx{D}_{\mx{h}_k}^H\mathbf{B}_k^H+\sigma_w^2\mx{I}_L\right)^{-1}\nonumber\\
    &\times\mx{B}_j\boldsymbol{\Sigma}_{\mx{r}_k}\mx{B}_j^H\left(P_p\mathbf{B}_k\mx{D}_{\mx{h}_k}\boldsymbol{\Sigma}_{\mathbf{g}_k}\mx{D}_{\mx{h}_k}^H\mathbf{B}_k^H+\sigma_w^2\mx{I}_L\right)^{-1}\mx{B}_k\mx{D}_{\mx{h}_k}\boldsymbol{\Sigma}_{\mx{g}_k}\nonumber\\
    &=\lim_{\sigma_w^2\to 0}\boldsymbol{\Sigma}_{\mx{g}_k}\mx{D}_{\mx{h}_k}^H\mx{B}_k^H\left(\mathbf{B}_k\mx{D}_{\mx{h}_k}\boldsymbol{\Sigma}_{\mathbf{g}_k}\mx{D}_{\mx{h}_k}^H\mathbf{B}_k^H+\frac{\sigma_w^2}{P_p}\mx{I}_L\right)^{-1}\nonumber\\
    &\times\mx{B}_j\boldsymbol{\Sigma}_{\mx{r}_k}\mx{B}_j^H\left(\mathbf{B}_k\mx{D}_{\mx{h}_k}\boldsymbol{\Sigma}_{\mathbf{g}_k}\mx{D}_{\mx{h}_k}^H\mathbf{B}_k^H+\frac{\sigma_w^2}{P_p}\mx{I}_L\right)^{-1}\mx{B}_k\mx{D}_{\mx{h}_k}\boldsymbol{\Sigma}_{\mx{g}_k}.
\label{eq:chanEstMSEHighSNR}
\end{align}
Note that since $\mx{B}_k^H\mx{B}_k = L\mx{I}_N$, the pseudoinverse corresponds to $\mx{B}_k^\dagger = \frac{1}{L}\mx{B}_k^H$. Consequently, we have
\begin{align}\nonumber
    &\lim_{\sigma_w^2\to 0}\mathbb{E}[(\mathbf{g}_k-\hat{\mathbf{g}}_k)(\mathbf{g}_k-\hat{\mathbf{g}}_k)^H]\\
    &=\frac{1}{L^4}\boldsymbol{\Sigma}_{\mx{g}_k}\mx{D}_{\mx{h}_k}^H\mx{B}_k^H\mx{B}_k\mx{D}_{\mx{h}_k}^{-H}\boldsymbol{\Sigma}_{\mathbf{g}_k}^{-1}\mx{D}_{\mx{h}_k}^{-1}\mathbf{B}_k^H\nonumber\\
    &\quad \times\mx{B}_j\boldsymbol{\Sigma}_{\mx{r}_k}\mx{B}_j^H\mx{B}_k\mx{D}_{\mx{h}_k}^{-H}\boldsymbol{\Sigma}_{\mathbf{g}_k}^{-1}\mx{D}_{\mx{h}_k}^{-1}\mathbf{B}_k^H\mx{B}_k\mx{D}_{\mx{h}_k}\boldsymbol{\Sigma}_{\mx{g}_k}\nonumber\\
    &=\frac{1}{L^2}\mx{D}_{\mx{h}_k}^{-1}\mathbf{B}_k^H\mx{B}_j\boldsymbol{\Sigma}_{\mx{r}_k}\mx{B}_j^H\mathbf{B}_k\mx{D}_{\mx{h}_k}^{-H}.\label{eq:pilotContaminationHighSNR}
\end{align}
From \eqref{eq:pilotContaminationHighSNR}, one can observe that the asymptotic behavior of the pilot contamination depends on the choice of $\mx{B}_1$ and $\mx{B}_2$. To analyze this, we consider two cases.
\subsubsection{Case 1: $\mx{B}_1=\mx{B}_2$}
In this case, we have that $\mx{B}_k^H\mx{B}_j=L\mx{I}_N$ for $j,k=1,2$ and $j\neq k$. Consequently, \eqref{eq:pilotContaminationHighSNR} becomes
\begin{align}\label{eq:case1Contamination}
    &\frac{1}{L^2}\mx{D}_{\mx{h}_k}^{-1}\mathbf{B}_k^H\mx{B}_j\boldsymbol{\Sigma}_{\mx{r}_k}\mx{B}_j^H\mathbf{B}_k\mx{D}_{\mx{h}_k}^{-H}=\mx{D}_{\mx{h}_k}^{-1}\boldsymbol{\Sigma}_{\mx{r}_k}\mx{D}_{\mx{h}_k}^{-H}.
\end{align}
Note that this result does not depend on $L$, which implies that this channel estimation error caused by pilot contamination cannot be eliminated by increasing the number of pilots when $\mx{B}_1 = \mx{B}_2$.
\subsubsection{The \acp{RIS} are configured such that $\mx{B}_1^H\mx{B}_2 = \boldsymbol{0}$}
In this case, $\mx{B}_k^H\mx{B}_j = \boldsymbol{0}$ for $k\neq j$ and $k,j\in\{1,2\}$. This implies that \eqref{eq:pilotContaminationHighSNR} is zero; thus,  \ac{BS} $k$ can estimate $\mx{g}_k$ even without being aware of $\mx{r}_k$. Also note that \eqref{eq:pilotContaminationHighSNR} implies that any choice of $\mx{B}_1$ and $\mx{B}_2$ that does not satisfy $\mx{B}_1^H\mx{B}_2=\boldsymbol{0}$ will result in pilot contamination while estimating correlated Rayleigh fading channels. We can summarize the high-\ac{SNR} behavior of the channel estimation error covariance matrix as
\begin{align}\nonumber
    \lim_{\sigma_w^2\to 0}\mathbb{E}[(\mathbf{g}_k-\hat{\mathbf{g}}_k)&(\mathbf{g}_k-\hat{\mathbf{g}}_k)^H]\\
    &=\begin{cases}
        \mx{D}_{\mx{h}_k}^{-1}\boldsymbol{\Sigma}_{\mx{r}_k}\mx{D}_{\mx{h}_k}^{-H} & \mx{B}_1 = \mx{B}_2,\\
        \boldsymbol{0} & \mx{B}_1^H\mx{B}_2=\boldsymbol{0}.\\
    \end{cases}\label{eq:pilotContaminationSummary}
\end{align}
This result shows that in order to estimate $\mx{g}_k$ reliably, it is necessary to configure the \acp{RIS} such that $\mx{B}_1^H\mx{B}_2 = \boldsymbol{0}$.

\section{Capacity Lower Bound for Reliable Communication under Imperfect \ac{CSI}}\label{sec:capBound}
In this section, we compute a lower bound on the ergodic capacity based on the imperfect \ac{CSI} obtained in the previous section via channel estimation. In particular, we consider the impact of pilot contamination and the effect of the signal model misspecification on the channel capacity. We derive the channel capacity lower bound for the two cases considered in Section \ref{sec:RayleighChan}.

\subsection{Capacity Lower Bound of a \ac{SISO} Channel with Channel Side Information}
Consider a generic \ac{SISO} system with the following received signal model:
\begin{equation}
    y = hx+w
\end{equation}
with $w\sim\mathcal{CN}(0,\sigma_w^2)$ and $x\sim\mathcal{CN}(0,1)$. Suppose that the receiver has partial information on $h$, denoted by $\Omega$. Then the capacity lower bound is \cite[Eq. 2.46]{MarzettaBook}
\begin{equation}\label{eq:capacityBound}
    C\geq\mathbb{E}_{\Omega}\left[\log_2\left(1+\frac{|\mathbb{E}[h|\Omega]|^2}{\text{Var}(h|\Omega)+\text{Var}(w|\Omega)}\right)\right].
\end{equation}
This bound is valid under certain conditions, which can be listed as follows \cite[Section 2.3.5]{MarzettaBook}:
\begin{itemize}
    \item The noise $w$ has zero-mean conditioned on $\Omega$, that is, $\mathbb{E}[w|\Omega]=0$.
    \item The transmitted signal $x$ and the noise $w$ are uncorrelated conditioned on $\Omega$, that is, $\mathbb{E}[xw^*|\Omega]=\mathbb{E}[x|\Omega]\mathbb{E}[w^*|\Omega]$.
    \item The received signal $hx$ and the noise $w$ are uncorrelated conditioned on $\Omega$, that is, $\mathbb{E}[hxw^*|\Omega]=\mathbb{E}[hx|\Omega]\mathbb{E}[w^*|\Omega]$.
\end{itemize}
In our setup, the data signal model for the two users can be expressed as
\begin{subequations}\label{eq:RayleighDataSignal}
    \begin{align}
        &y_1 = \sqrt{P_d}(\mx{h}_1^T\boldsymbol{\Phi}_1\mx{g}_1 + \mx{q}_1^T\boldsymbol{\Phi}_2\mx{p}_1)x_1+w_1,\\
        &y_2 = \sqrt{P_d}(\mx{h}_2^T\boldsymbol{\Phi}_2\mx{g}_2 + \mx{q}_2^T\boldsymbol{\Phi}_1\mx{p}_2)x_2+w_2,
    \end{align}
\end{subequations}
where $x_k\sim\mathcal{CN}(0,1)$ denotes the transmitted data for $k=1,2$ and $P_d$ denotes the data transmission power. During the data transmission phase, both \acp{RIS} are configured to phase-align the cascaded channel, that is:
\begin{equation}\label{eq:RayleighRISConfig}
    \phi_{kn} = \arg(h_{kn})+\arg(\hat{g}_{kn}).
\end{equation}
For both \acp{BS}, we consider the side information $\Omega$ as the knowledge of $\boldsymbol{\Phi}_1$, $\boldsymbol{\Phi}_2$, $\mx{h}_1$, $\mx{h}_2$, $\hat{\mx{g}}_1$, and $\hat{\mx{g}}_2$. Consequently, the outer expectation in \eqref{eq:capacityBound} refers to the expectation with respect to the marginal distributions of $\hat{\mx{g}}_1$ and $\hat{\mx{g}}_2$.
\begin{lem}
    The system setup described by the signal model in \eqref{eq:RayleighDataSignal} that uses the \ac{RIS} configurations described in \eqref{eq:RayleighRISConfig} satisfies the three regularity conditions required by the capacity bound in \eqref{eq:capacityBound}.
\end{lem}
\begin{proof}
    First, let us identify the $h$ and $w$ that we had defined in our system setup in \eqref{eq:RayleighDataSignal}:
    \begin{subequations}
        \begin{align}
            &h = \sqrt{P_d}(\mx{h}_k^T\boldsymbol{\Phi}_k\mx{g}_k+\mx{q}_k^T\boldsymbol{\Phi}_j\mx{p}_k)\\
            &w = w_k
        \end{align}
    \end{subequations}
    We can now prove that the conditional mean of the noise conditioned on the channel side information, which we can describe as $\Omega=\hat{\mx{g}}_k$, is zero. Note that $\hat{\mx{g}}_k$ and $w_k$ are independent: for $\mx{B}_1=\mx{B}_2$, the less trivial case, we have that $\hat{\mx{g}_k}=\mx{g}_k+\mx{D}_{\mx{h}_k}^{-1}\mx{r}_k$, that is, no dependence on the noise. At lower \ac{SNR}s, the vanishing components contain the realizations of the noise received during channel estimation, and considering the fact that the receiver noise is white over time, that does not affect the independence between $\hat{\mx{g}}_k$ and $w_k$ either. Therefore, $\mathbb{E}[w_k|\hat{\mx{g}}_k] = \mathbb{E}[w_k]=0$.

    It is also straightforward to prove that the transmitted signal is uncorrelated with the noise conditioned on $\Omega$ due to the fact that $x_k$ and $w_k$ are independent of $\Omega$ individually. So the expression $\mathbb{E}[x_kw_k^*|\Omega]$ does not have anything that depends on $\Omega$, i.e., $\mathbb{E}[x_kw_k^*|\Omega] = \mathbb{E}[x_kw_k^*]=\mathbb{E}[x_k]\mathbb{E}[w_k^*]=\mathbb{E}[x_k|\Omega]\mathbb{E}[w_k^*|\Omega]$ can be obtained, proving that $x$ and $w$ are uncorrelated with each other conditioned on $\Omega$.

    The last point is also quite straightforward since $\Omega$ does not contain any relations between the overall channel, the receiver noise, and the transmitted signal. Therefore it is easy to claim that this regularity condition also holds, hence the capacity bound provided in \eqref{eq:capacityBound} is applicable to our system setup.
\end{proof}
\subsection{Capacity Lower Bound with High-\ac{SNR} Channel Estimates Available at the \acp{BS}}
If we assume that the channel estimation is performed at a high \ac{SNR}, we can model the channel estimation error according to \eqref{eq:pilotContaminationSummary}. Consequently, we can express the channel estimates in terms of the true channels and the channel estimation error as
\begin{align}\label{eq:chanEst_ErrModel}
    \hat{\mx{g}}_k = \mx{g}_k + \mx{e}_k,
\end{align}
where the channel estimation error is
\begin{align}
    \mx{e}_k = \begin{cases}
        \boldsymbol{0} & \mx{B}_1^H\mx{B}_2=\boldsymbol{0},\\
        \mx{D}_{\mx{h}_k}^{-1}\mx{r}_k & \mx{B}_1=\mx{B}_2.
    \end{cases}\label{eq:CSI_error}
\end{align}
Hence, \ac{BS} $k$ knows $\mx{g}_k$ perfectly if the \acp{RIS} are configured such that $\mx{B}_1^H\mx{B}_2 = \boldsymbol{0}$ during channel estimation, and $\mx{e}_k=\mx{D}_{\mx{h}_k}^{-1}\mx{r}_k$ when $\mx{B}_1=\mx{B}_2$. Also note that even when we have $\mx{B}_1=\mx{B}_2$, the channel and the channel estimation error are independent. Consequently, we can rewrite the overall \ac{SISO} channel as
\begin{equation}
    v_k\triangleq \sqrt{P_d}(\boldsymbol{\phi}_k^T\mx{D}_{\mx{h}_k}\mx{g}_k + \boldsymbol{\phi}_j^T\mx{r}_k).
\end{equation}
The mean of the overall \ac{SISO} channel conditioned on the side information can be expressed as
\begin{align}
    \mathbb{E}[v_k|\Omega] &= \sqrt{P_d}(\boldsymbol{\phi}_k^T\mx{D}_{\mx{h}_k}\mathbb{E}[\mx{g}_k|\hat{\mx{g}}_k] + \boldsymbol{\phi}_j^T\mathbb{E}[\mx{r}_k|\hat{\mx{g}}_k])
\end{align}
Here, we can utilize the channel estimate structure provided by \eqref{eq:chanEst_ErrModel} and \eqref{eq:CSI_error}. Note that when we consider $\mathbb{E}[\mx{g}_k|\hat{\mx{g}}_k]$, it can be thought as estimating $\mx{g}_k$ based on observing $\hat{\mx{g}}_k$ since both $\mx{g}_k$ and $\mx{e}_k$ and they are independent from each other. The same goes for computing $\mathbb{E}[\mx{r}_k|\hat{\mx{g}}_k]$. Since the $\mx{g}_k-\hat{\mx{g}}_k$ and $\mx{r}_k-\hat{\mx{g}}_k$ are jointly Gaussian, the \ac{MMSE} estimate, also known as the conditional mean estimate coincides with the \ac{LMMSE} estimate, therefore, we can use the \ac{LMMSE} formulation here:
\begin{subequations}
    \begin{align}
        &\mathbb{E}[\mx{g}_k|\hat{\mx{g}}_k] = \mathbb{E}[\mx{g}_k\hat{\mx{g}}_k^H](\mathbb{E}[\hat{\mx{g}}_k\hat{\mx{g}}_k^H])^{-1}\hat{\mx{g}}_k\\
        &\mathbb{E}[\mx{r}_k|\hat{\mx{g}}_k] = \mathbb{E}[\mx{r}_k\hat{\mx{g}}_k^H](\mathbb{E}[\hat{\mx{g}}_k\hat{\mx{g}}_k^H])^{-1}\hat{\mx{g}}_k
    \end{align}
\end{subequations}
which can be expressed more explicitly as
\begin{subequations}\label{eq:condExp}
    \begin{align}
        &\mathbb{E}[\mx{g}_k|\hat{\mx{g}}_k] = \boldsymbol{\Sigma}_{\mx{g}_k}(\boldsymbol{\Sigma}_{\mx{g}_k}+\mx{D}_{\mx{h}_k}^{-1}\boldsymbol{\Sigma}_{\mx{r}_k}\mx{D}_{\mx{h}_k}^{-H})^{-1}\hat{\mx{g}}_k,\\
        &\mathbb{E}[\mx{r}_k|\hat{\mx{g}}_k] = \nonumber\\&\begin{cases}
            \boldsymbol{\Sigma}_{\mx{r}_k}\mx{D}_{\mx{h}_k}^{-H}(\boldsymbol{\Sigma}_{\mx{g}_k}+\mx{D}_{\mx{h}_k}^{-1}\boldsymbol{\Sigma}_{\mx{r}_k}\mx{D}_{\mx{h}_k}^{-H})^{-1}\hat{\mx{g}}_k & \mx{B}_1=\mx{B}_2,\\
            \boldsymbol{0} & \mx{B}_1^H\mx{B}_2 = \boldsymbol{0}.
        \end{cases}        
    \end{align}
\end{subequations}
On the other hand, the variance of $v_k$ conditioned on $\hat{\mx{g}}_k$ can be expressed as
\begin{align}\nonumber
    &\text{Var}(v_k|\hat{\mx{g}}_k)=P_d\boldsymbol{\phi}_k^T\mx{D}_{\mx{h}_k}\text{Var}(\mx{g}_k|\hat{\mx{g}}_k)\mx{D}_{\mx{h}_k}^H\boldsymbol{\phi}_k^*\\
    &+P_d\boldsymbol{\phi}_j^T\text{Var}(\mx{r}_k|\hat{\mx{g}}_k)\boldsymbol{\phi}_j^*+2\text{Re}(\boldsymbol{\phi}_k^T\mx{D}_{\mx{h}_k}\mathbb{E}[\mx{g}_k\mx{r}_k^H|\hat{\mx{g}}_k]\boldsymbol{\phi}_j^*)\nonumber\\
    &-2\text{Re}(\boldsymbol{\phi}_k^T\mx{D}_{\mx{h}_k}\mathbb{E}[\mx{g}_k|\hat{\mx{g}}_k]\mathbb{E}[\mx{r}_k^H|\hat{\mx{g}}_k]\boldsymbol{\phi}_j^*)\label{eq:condVar_v1}
\end{align}
where $\mathbb{E}[\mx{g}_k|\hat{\mx{g}}_k]$ and $\mathbb{E}[\mx{r}_k|\hat{\mx{g}}_k]$ are provided by \eqref{eq:condExp}. In addition, we can use \ac{LMMSE} formulation results for $\text{Var}(\mx{g}_k|\hat{\mx{g}}_k)$ and $\text{Var}(\mx{r}_k|\hat{\mx{g}}_k)$ which correspond to the error covariance matrices as a result of estimating $\mx{g}_k$ and $\mx{r}_k$ with an \ac{LMMSE} estimator based on the observation $\hat{\mx{g}}_k$. Consequently, these two terms can be expressed as
\begin{subequations}
    \begin{align}
        &\text{Var}(\mx{g}_k|\hat{\mx{g}}_k) =\mathbb{E}[\mx{g}_k\mx{g}_k^H]-\mathbb{E}[\mx{g}_k\hat{\mx{g}}_k^H](\mathbb{E}[\mx{g}_k\hat{\mx{g}}_k^H])^{-1}\mathbb{E}[\hat{\mx{g}}_k\mx{g}_k^H],\\
        &\text{Var}(\mx{r}_k|\hat{\mx{g}}_k) =\mathbb{E}[\mx{r}_k\mx{r}_k^H]-\mathbb{E}[\mx{r}_k\hat{\mx{g}}_k^H](\mathbb{E}[\mx{r}_k\hat{\mx{g}}_k^H])^{-1}\mathbb{E}[\hat{\mx{g}}_k\mx{r}_k^H].
    \end{align}
\end{subequations}
Computing the expectations above and also the cross-term $\mathbb{E}[\mx{g}_k\mx{r}_k^H|\hat{\mx{g}}_k]$, we can obtain the implicit expressions in \eqref{eq:condVar_v1} as
\begin{subequations}
    \begin{align}
        &\text{Var}(\mx{g}_k|\hat{\mx{g}}_k) =\boldsymbol{\Sigma}_{\mx{g}_k}-\boldsymbol{\Sigma}_{\mx{g}_k}(\boldsymbol{\Sigma}_{\mx{g}_k}+\mx{D}_{\mx{h}_k}^{-1}\boldsymbol{\Sigma}_{\mx{r}_k}\mx{D}_{\mx{h}_k}^{-H})^{-1}\boldsymbol{\Sigma}_{\mx{g}_k}\\
        &\text{Var}(\mx{r}_k|\hat{\mx{g}}_k) \nonumber\\&=\boldsymbol{\Sigma}_{\mx{r}_k}- \boldsymbol{\Sigma}_{\mx{r}_k}\mx{D}_{\mx{h}_k}^{-H}(\boldsymbol{\Sigma}_{\mx{g}_k}+\mx{D}_{\mx{h}_k}^{-1}\boldsymbol{\Sigma}_{\mx{r}_k}\mx{D}_{\mx{h}_k}^{-H})^{-1}\mx{D}_{\mx{h}_k}^{-1} \boldsymbol{\Sigma}_{\mx{r}_k}\\
        &\mathbb{E}[\mx{g}_k\mx{r}_k^H|\hat{\mx{g}}_k]\nonumber\\&= \hat{\mx{g}}_k\hat{\mx{g}}_k^H(\boldsymbol{\Sigma}_{\mx{g}_k}+\mx{D}_{\mx{h}_k}^{-1}\boldsymbol{\Sigma}_{\mx{r}_k}\mx{D}_{\mx{h}_k}^{-H})^{-1}\boldsymbol{\Sigma}_{\mx{r}_k}\mx{D}_{\mx{h}_k}^{-H}-\mx{D}_{\mx{h}_k}^{-1}\boldsymbol{\Sigma}_{\mx{r}_k}
    \end{align}
\end{subequations}
for $\mx{B}_1=\mx{B}_2$ and
\begin{subequations}
    \begin{align}
        &\text{Var}(\mx{g}_k|\hat{\mx{g}}_k)=\boldsymbol{0}\\
        &\text{Var}(\mx{r}_k|\hat{\mx{g}}_k)=\boldsymbol{\Sigma}_{\mx{r}_k}\\
        &\mathbb{E}[\mx{g}_k\mx{r}_k^H|\hat{\mx{g}}_k]=\boldsymbol{0}
    \end{align}
\end{subequations}
for $\mx{B}_1^H\mx{B}_2=\boldsymbol{0}$. Obtaining $\text{Var}(w_k|\Omega)=\sigma_w^2$ is straightforward since $w_k$ is independent from $\Omega$. Also note that $\hat{\mx{g}}_k$ depends on the choice of $\mx{B}_1$ and $\mx{B}_2$. As a result, the capacity lower bound can be expressed as
\begin{align}
    C_k\geq \mathbb{E}_\Omega\left[\log_2\left(1+\frac{|\mathbb{E}[v_k|\Omega]|^2}{\text{Var}(v_k|\Omega)+\sigma_w^2}\right)\right]
\end{align}
with $\mathbb{E}[v_k|\Omega]$ and $\text{Var}(v_k|\Omega)$ taking values according to the choice of $\mx{B}_1$ and $\mx{B}_2$.

\section{Numerical Results}\label{sec:NumRes}
In this section, we provide numerical examples to demonstrate the implications of the analytical results obtained in Sections \ref{sec:MLE}-\ref{sec:capBound}. First, we provide the numerical results for the channel and data estimation \acp{MSE} when deterministic channels are considered. For correlated Rayleigh fading, we demonstrate the impact of pilot contamination on channel estimation \ac{MSE} and the resulting capacity lower bound.

\subsection{Estimation Performance With Deterministic Channels}\label{sec:NumResChanEstDet}

For deterministic channel estimation, we consider the \ac{NMSE} as our performance metric. Moreover, we consider the results for a single \ac{UE}, since the results for different \acp{UE} only differ by the channel realizations. For the deterministic channel $\mx{g}_k$, we obtain the \ac{NMSE} as
\begin{equation}
    \text{NMSE} = \frac{\text{MSE}}{\|\mx{g}_k\|^2}.
\end{equation}
\begin{table}[t]
    \centering
    \caption{Parameters used in Figures \ref{fig:chanEstSNR} and \ref{fig:dataMSESNR}.}
    \begin{tabular}{|c|c|}
        \hline
        Parameter & Value\\
        \hline
        $P_p$ or $P_d$ & $-30,-25,\dots,40$ dBm\footnote{The results for $P_p = 45,50,55,\text{ and }60$ dBm are also demonstrated in Fig. \ref{fig:chanEstSNR} to display the high \ac{SNR} floor more clearly.}\\
        \hline
        \ac{UE}-\ac{RIS} path loss & $-80$ dB\\
        \hline
        \ac{RIS}-\ac{BS} path loss & $-60$ dB\\
        \hline
        $\sigma_w^2$ & $-90$ dBm\\
        \hline
        $N$ & $256$\\
        \hline
        $L$ & $513$\\
        \hline
    \end{tabular}
    \label{tab:numResParams}
\end{table}
In Fig. \ref{fig:chanEstSNR}, we plot \eqref{eq:traceError} for different values of $P_p$, and we also provide the high-\ac{SNR} floor for the case where $\mx{B}_1=\mx{B}_2$. The set of parameters used to generate Fig.~\ref{fig:chanEstSNR} and \ref{fig:dataMSESNR} are provided in Table~\ref{tab:numResParams}. On the other hand, the range of transmission power in Fig.~\ref{fig:capBound} is between $-10$ and $60$ dBm. In addition, for Fig.~\ref{fig:rayleighChanEst}, \ref{fig:rayleighChanEst} and \ref{fig:rayleighChanEst}, we consider a $8\times 8$ \ac{URA} geometry with $\lambda/2$ spacing in both vertical and horizontal axes. Therefore, the parameter values $N=64$ and $L=128$ apply to those figures. Note that at lower transmission powers, the covariance matrix of the estimator acts dominantly, hence, both \ac{RIS} configurations perform nearly the same. However, after $P_p=20$ dBm, the power of the estimator bias starts to dominate, and the average \ac{MSE} for $\mx{B}_1=\mx{B}_2$ converges to the floor denoted by the black dashed line, which is given by \eqref{eq:chEstFloor}. On the other hand, the average \ac{MSE} for $\mx{B}_1^H\mx{B}_2=\boldsymbol{0}$ does not stop there but keeps decreasing towards zero. As mentioned before, the \ac{MML} estimators used by the \acp{BS} coincide with the true \ac{ML} estimators when the \acp{RIS} are configured such that $\mx{B}_1^H\mx{B}_2=\boldsymbol{0}$.

\begin{figure}
    \centering
    \includegraphics[width=\linewidth]{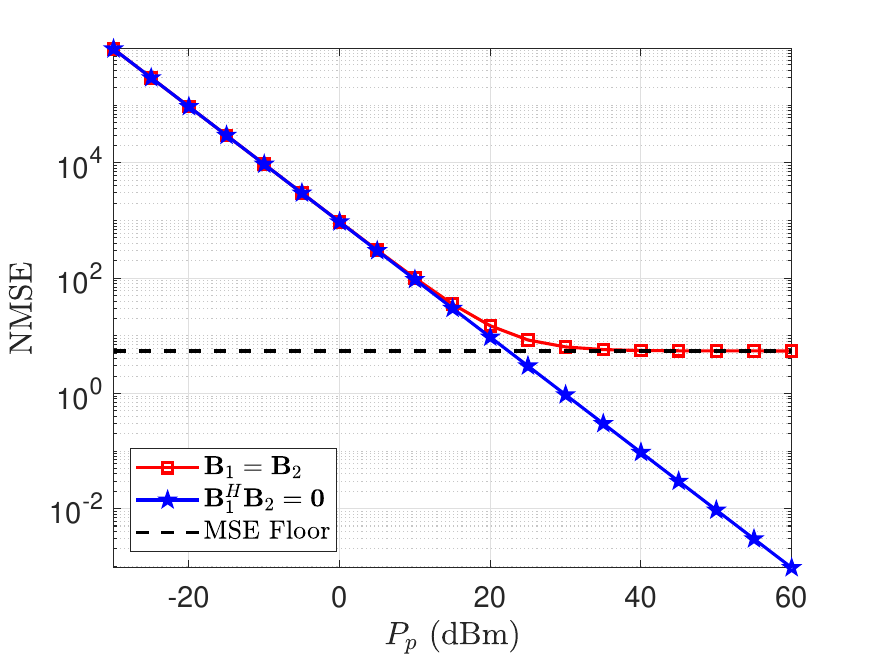}
    \caption{Pilot transmission power versus the channel estimation \ac{NMSE} for deterministic channels. Since the prior distribution of the parameter vector is not considered in the non-random parameter estimation framework, it is highly likely to obtain \acp{NMSE} greater than $1$.}
    \label{fig:chanEstSNR}
\end{figure}

\subsection{Data Estimation with Deterministic Channels}\label{sec:NumResDataEst}
In Fig. \ref{fig:dataMSESNR}, the data estimation \ac{MSE} performance with the two \ac{RIS} pilot configurations is analyzed when the channel estimation \ac{SNR} is high as in Section \ref{sec:highSNRChanEstRef}. That is, \eqref{eq:dataMSE_highChan} is plotted for $\mx{B}_1=\mx{B}_2$ and $\mx{B}_1^H\mx{B}_2=\boldsymbol{0}$. In addition, the case where all of the channels are perfectly known is plotted to serve as the golden standard, labeled as \emph{Perfect \ac{CSI}}. However, even when all the channels are perfectly known, each \ac{RIS} is assumed to be optimized independently according to the subscribed \ac{UE}'s \ac{CSI}. Note that although the channel estimation \ac{SNR} is high, $\mx{B}_1=\mx{B}_2$ yields biased estimates of $\mx{g}_1$ due to pilot contamination caused by self-interference. On the other hand, $\mx{B}_1^H\mx{B}_2=\boldsymbol{0}$ yields the true $\mx{g}_1$ as the estimate, however, since \ac{BS} $1$ is unaware of the path through the second \ac{RIS}, the data estimate is biased, hence, there is still a high data transmission \ac{SNR} floor. At around $P_d=5$ dBm, $\mx{B}_1=\mx{B}_2$ starts to approach the high-\ac{SNR} floor. On the other hand, $\mx{B}_1^H\mx{B}_2=\boldsymbol{0}$ does not suffer from the lack of awareness of the second \ac{RIS} path until around $P_d=20$ dBm. Hence, Fig. \ref{fig:dataMSESNR} clearly shows the benefit of configuring the \ac{RIS} pilot configurations sequences orthogonally. 

\begin{figure}
    \centering
    \includegraphics[width=\linewidth]{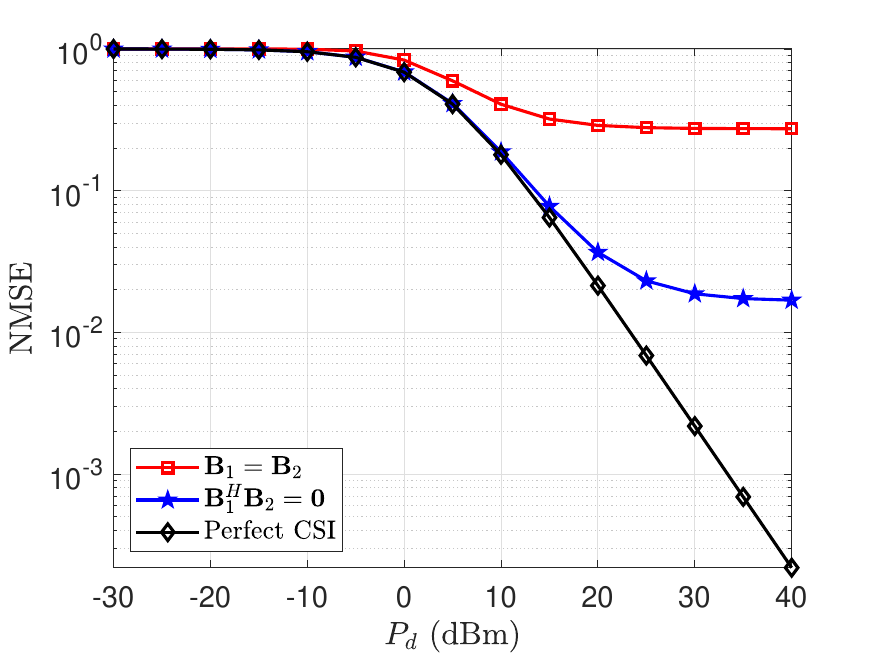}
    \caption{Data transmission power versus the data estimation \ac{NMSE} for deterministic channels with high channel estimation \ac{SNR}.}
    \label{fig:dataMSESNR}
\end{figure}

\subsection{Channel Estimation Based on Correlated Rayleigh Fading Priors}\label{sec:numResChanEstRay}
In Fig.~\ref{fig:rayleighChanEst}, the channel estimation performance with the two \ac{RIS} pilot configurations are analyzed for correlated Rayleigh fading channels. Fig.~\ref{fig:rayleighChanEst} is generated by computing the trace of the channel estimation error covariance matrix provided in \eqref{eq:chanEstMSE} and normalizing it by the factor of $\text{tr}(\boldsymbol{\Sigma}_{\mx{g}_k})$ for different $P_p$ values and different spatial channel correlation matrices. The spatial channel correlation matrices are computed for isotropic scattering based on the different \ac{RIS} element geometries according to \cite[Prop. 1]{Bjornson2021spatial}.
\begin{figure}
    \centering
    \includegraphics[width=\linewidth]{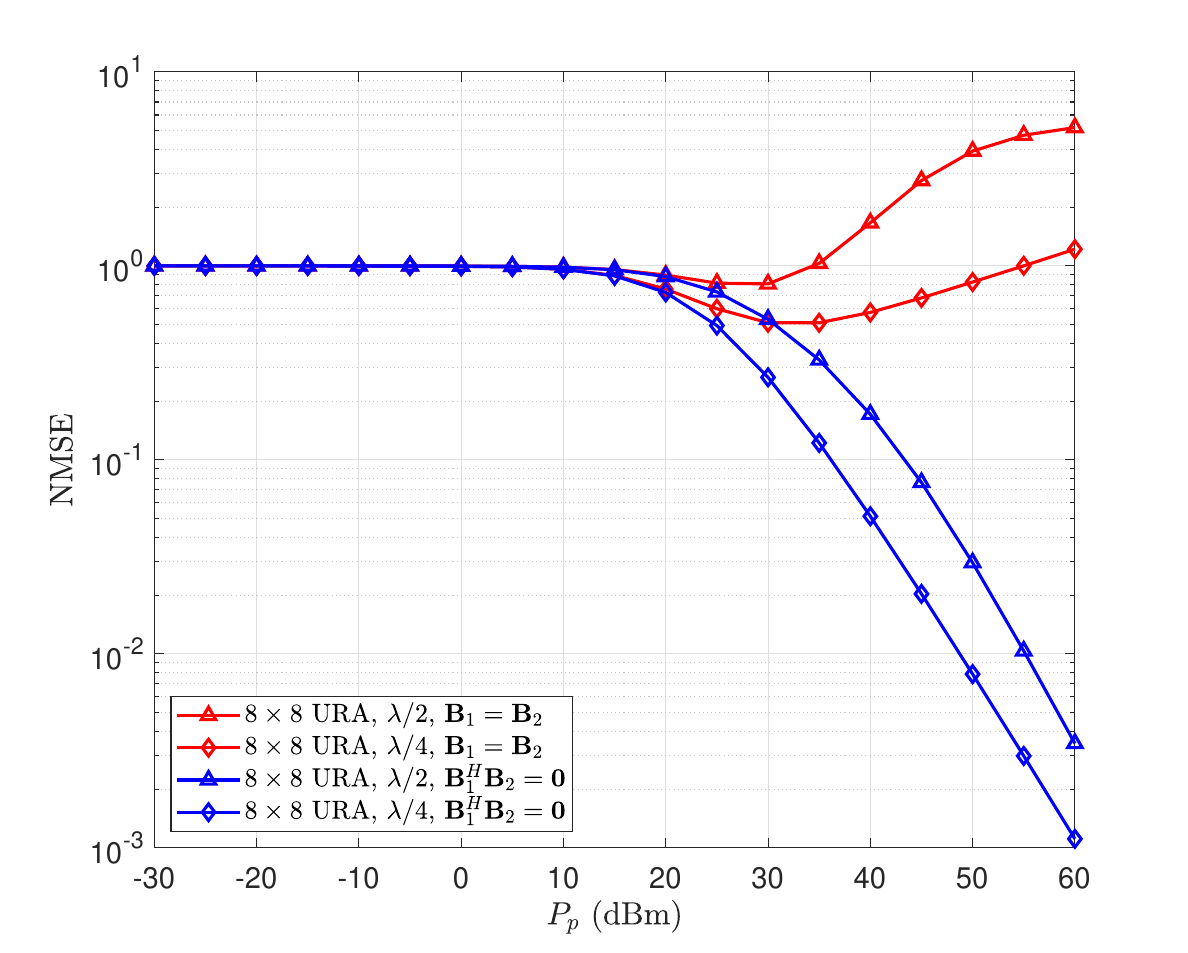}
    \caption{Pilot transmission power versus channel estimation \ac{MSE} for different \ac{RIS} geometries.}
    \label{fig:rayleighChanEst}
\end{figure}
Note that configuring the \acp{RIS} such that $\mx{B}_1=\mx{B}_2$ causes severe problems in channel estimation, that is, the \ac{MSE} increases as pilot transmission power increases for all geometries while $\mx{B}_1^H\mx{B}_2=\boldsymbol{0}$ completely eliminates pilot contamination. In addition, one can note that as the spatial correlation increases, channel estimation performance also increases since different channel parameters contain more information from one another.

Furthermore, in Fig.~\ref{fig:rayleighChanEst_analysis}, we demonstrate the two different components of the channel estimation error in the presence of inter-operator pilot contamination. While the red curve corresponds to the total \ac{NMSE} and the blue curve corresponds to the \ac{NMSE} in the absence of pilot contamination as usual, the green curve demonstrates the term coming from pilot contamination, as in \eqref{eq:chanEstMSE}. On the other hand, the black dashed line represents the asymptote of the channel estimation \ac{NMSE} in the presence of pilot contamination, which is provided by \eqref{eq:pilotContaminationSummary}. Note that as the transmission power increases, the \ac{NMSE} coming from pilot contamination also increases up to a certain point, and converges to the trace of \eqref{eq:pilotContaminationSummary} for $\mx{B}_1=\mx{B}_2$ due to the fact that at very high transmission powers, the increase in pilot contamination cancels out with the increasing ability to estimating the channel.
\begin{figure}
    \centering
    \includegraphics[width=\linewidth]{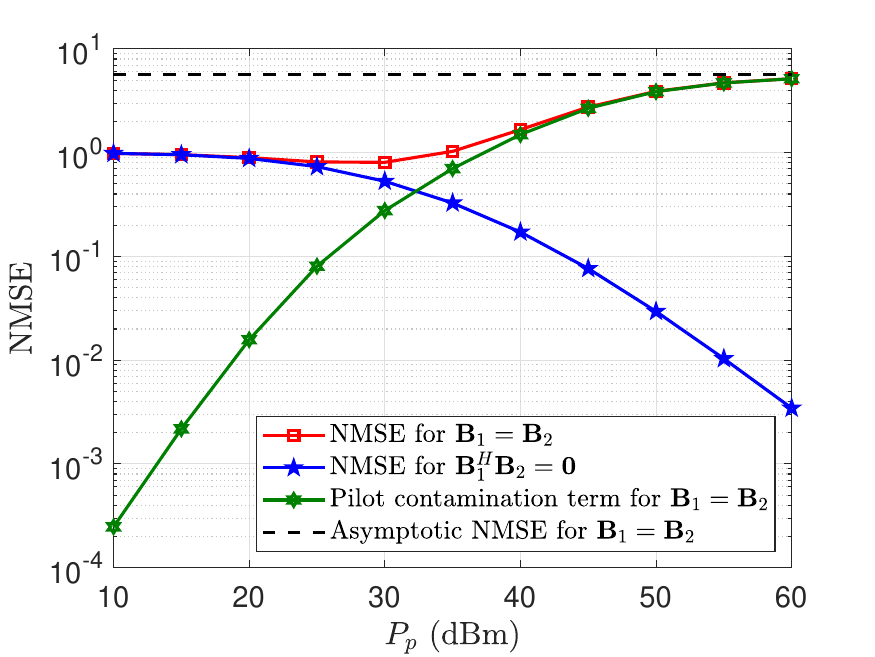}
    \caption{Different components of the channel estimation \ac{NMSE}}
    \label{fig:rayleighChanEst_analysis}
\end{figure}

\subsection{Capacity Lower Bound for Reliable Communication Under Imperfect \ac{CSI}}
In Fig.~\ref{fig:capBound}, the capacity lower bound derived in Section \ref{sec:capBound} is plotted against the data transmission power. This is performed by generating several channel realizations and computing \eqref{eq:capacityBound}. Note that when $\mx{B}_1=\mx{B}_2$, the capacity lower bound stops increasing after $P_d=30$ dBm while this happens at around $P_d=40$ dBm for $\mx{B}_1^H\mx{B}_2=\boldsymbol{0}$ when the effect of the misspecified channel during data transmission starts to appear. In any case, it is clear that configuring the \acp{RIS} such that $\mx{B}_1^H\mx{B}_2=\boldsymbol{0}$ almost doubles the capacity lower bound. 
\begin{figure}
    \centering
    \includegraphics[width=\linewidth]{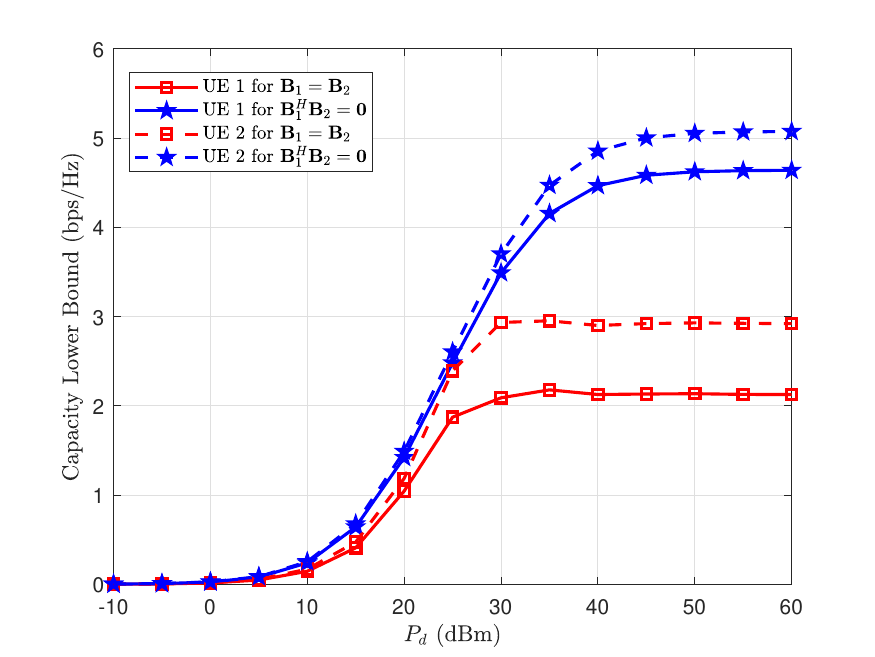}
    \caption{Data transmission power versus capacity lower bound for the two users.}
    \label{fig:capBound}
\end{figure}
\section{Conclusions}\label{sec:Conclusion}
In this paper, we have studied the impact of pilot contamination in a system consisting of two wide-band \acp{RIS}, two single-antenna \acp{UE}, and two co-located single-antenna \acp{BS}. We have demonstrated that the presence of multiple \acp{RIS} in the same area causes pilot contamination, although the \acp{UE} are subscribed to different operators and transmit over disjoint narrow frequency bands. 
To combat this new type of pilot contamination, we have proposed the use of orthogonal \ac{RIS} configurations during pilot transmission. For two different sets of assumptions, that is, deterministic and correlated Rayleigh-fading channel models, we have derived the channel and data estimation \acp{MSE} and the capacity lower bound in closed-form. In the numerical results, we have clearly shown that the proposed approach eliminates pilot contamination completely, and decreases data estimation \ac{MSE} significantly for deterministic channels. On the other hand, we have also shown that the capacity lower bound almost doubles when the \acp{RIS} are configured orthogonally during the pilot transmission step. While one might argue that this doubling comes at the expense of doubling the number of pilots, the estimates can be used for many data transmissions if the channel is static enough, resulting in a higher overall data rate. While this study covers the channel estimation performance in multi-operator \ac{RIS}-based pilot contamination scenarios for both deterministic and stochastic channels, further analysis is needed for parametric channel models, 
which opens a new set of possibilities.
\bibliographystyle{IEEEtran}
\bibliography{IEEEabrv,references}

\begin{thebibliography}{10}
\providecommand{\url}[1]{#1}
\csname url@samestyle\endcsname
\providecommand{\newblock}{\relax}
\providecommand{\bibinfo}[2]{#2}
\providecommand{\BIBentrySTDinterwordspacing}{\spaceskip=0pt\relax}
\providecommand{\BIBentryALTinterwordstretchfactor}{4}
\providecommand{\BIBentryALTinterwordspacing}{\spaceskip=\fontdimen2\font plus
\BIBentryALTinterwordstretchfactor\fontdimen3\font minus
  \fontdimen4\font\relax}
\providecommand{\BIBforeignlanguage}[2]{{%
\expandafter\ifx\csname l@#1\endcsname\relax
\typeout{** WARNING: IEEEtran.bst: No hyphenation pattern has been}%
\typeout{** loaded for the language `#1'. Using the pattern for}%
\typeout{** the default language instead.}%
\else
\language=\csname l@#1\endcsname
\fi
#2}}
\providecommand{\BIBdecl}{\relax}
\BIBdecl

\bibitem{Gurgunoglu2023_BlackSeaCom}
D.~G\"{u}rg\"{u}no\u{g}lu, E.~Bj\"{o}rnson, and G.~Fodor, ``Impact of pilot
  contamination between operators with interfering reconfigurable intelligent
  surfaces,'' 2023.

\bibitem{Marzetta2010a}
T.~L. Marzetta, ``Noncooperative cellular wireless with unlimited numbers of
  base station antennas,'' \emph{IEEE Trans. Wireless Commun.}, vol.~9, no.~11,
  pp. 3590--3600, 2010.

\bibitem{Sanguinetti2020a}
L.~Sanguinetti, E.~Bj{\"o}rnson, and J.~Hoydis, ``Toward {Massive MIMO 2.0}:
  Understanding spatial correlation, interference suppression, and pilot
  contamination,'' \emph{IEEE Trans. Commun.}, vol.~68, no.~1, 2020.

\bibitem{pCon1}
J.~Jose, A.~Ashikhmin, T.~L. Marzetta, and S.~Vishwanath, ``Pilot contamination
  and precoding in multi-cell {TDD} systems,'' \emph{{IEEE} Trans. Wireless
  Commun.}, vol.~10, no.~8, pp. 2640--2651, 2011.

\bibitem{Saxena:15}
V.~Saxena, G.~Fodor, and E.~Karipidis, ``Mitigating pilot contamination by
  pilot reuse and power control schemes for massive {MIMO} systems,'' in
  \emph{2015 IEEE 81st Vehicular Technology Conference (VTC Spring)}, 2015, pp.
  1--6.

\bibitem{Fodor:17}
G.~Fodor, N.~Rajatheva, W.~Zirwas, L.~Thiele, M.~Kurras, K.~Guo, A.~Tolli,
  J.~H. Sorensen, and E.~De~Carvalho, ``An overview of massive {MIMO}
  technology components in {METIS},'' \emph{IEEE Communications Magazine},
  vol.~55, no.~6, pp. 155--161, 2017.

\bibitem{ris_commag}
C.~Pan, H.~Ren, K.~Wang, J.~F. Kolb, M.~Elkashlan, M.~Chen, M.~Di~Renzo,
  Y.~Hao, J.~Wang, A.~L. Swindlehurst, X.~You, and L.~Hanzo, ``Reconfigurable
  intelligent surfaces for {6G} systems: Principles, applications, and research
  directions,'' \emph{IEEE Communications Magazine}, vol.~59, no.~6, pp.
  14--20, 2021.

\bibitem{Bjornson2020a}
E.~Bj\"ornson, {\"O}.~\"Ozdogan, and E.~G. Larsson, ``Reconfigurable
  intelligent surfaces: Three myths and two critical questions,'' \emph{{IEEE}
  Commun. Mag.}, no.~12, pp. 90--96, 2020.

\bibitem{Araujo:23}
G.~T. de~Ara\'{u}jo, A.~L.~F. de~Almeida, R.~Boyer, and G.~Fodor, ``Semi-blind
  joint channel and symbol estimation for {IRS}-assisted {MIMO} systems,''
  \emph{IEEE Transactions on Signal Processing}, vol.~71, pp. 1184--1199, 2023.

\bibitem{Liu:23}
R.~Liu, M.~Li, H.~Luo, Q.~Liu, and A.~L. Swindlehurst, ``Integrated sensing and
  communication with reconfigurable intelligent surfaces: {Opportunities},
  applications, and future directions,'' \emph{IEEE Wireless Communications},
  vol.~30, no.~1, pp. 50--57, 2023.

\bibitem{risChanEst}
L.~Wei, C.~Huang, G.~C. Alexandropoulos, C.~Yuen, Z.~Zhang, and M.~Debbah,
  ``Channel estimation for {RIS}-empowered multi-user {MISO} wireless
  communications,'' \emph{IEEE Transactions on Communications}, vol.~69, no.~6,
  pp. 4144--4157, 2021.

\bibitem{Bjornson2022b}
E.~Bj\"ornson and P.~Ramezani, ``Maximum likelihood channel estimation for
  {RIS}-aided communications with {LOS} channels,'' in \emph{Asilomar
  Conference on Signals, Systems and Computers}, 2022.

\bibitem{ris_zappone}
C.~Huang, A.~Zappone, G.~C. Alexandropoulos, M.~Debbah, and C.~Yuen,
  ``Reconfigurable intelligent surfaces for energy efficiency in wireless
  communication,'' \emph{{IEEE} Trans. Commun.}, vol.~18, no.~8, pp.
  4157--4170, 2019.

\bibitem{Zhang2023}
Z.~Zhang, L.~Dai, X.~Chen, C.~Liu, F.~Yang, R.~Schober, and H.~V. Poor,
  ``Active {RIS} vs. passive {RIS}: Which will prevail in 6g?'' \emph{IEEE
  Transactions on Communications}, vol.~71, no.~3, pp. 1707--1725, 2023.

\bibitem{Rihan:23}
M.~Rihan, A.~Zappone, S.~Buzzi, G.~Fodor, and M.~Debbah, ``Passive vs. active
  reconfigurable intelligent surfaces for integrated sensing and communication:
  {Challenges} and opportunities,'' \emph{IEEE Network}, pp. 1--1, 2023.

\bibitem{Garg:22}
N.~Garg, H.~Ge, and T.~Ratnarajah, ``Generalized superimposed training scheme
  in {IRS}-assisted cell-free massive mimo systems,'' \emph{IEEE Journal of
  Selected Topics in Signal Processing}, vol.~16, no.~5, pp. 1157--1171, 2022.

\bibitem{38211}
3GPP, ``{NR; Physical channels and modulation},'' {3rd Generation Partnership
  Project (3GPP)}, Technical Specification (TS) {38.211}, 09 2022, version
  17.4.0.

\bibitem{Zhang:20}
S.~Zhang and R.~Zhang, ``Capacity characterization for intelligent reflecting
  surface aided {MIMO} communication,'' \emph{IEEE Journal on Selected Areas in
  Communications}, vol.~38, no.~8, pp. 1823--1838, 2020.

\bibitem{Pei:21}
X.~Pei, H.~Yin, L.~Tan, L.~Cao, Z.~Li, K.~Wang, K.~Zhang, and E.~Bj\"ornson,
  ``{RIS}-aided wireless communications: {Prototyping}, adaptive beamforming,
  and indoor/outdoor field trials,'' pp. 8627--8640, 2021.

\bibitem{Tampio:21}
V.~Tapio, A.~Shojaeifard, I.~Hemadeh, A.~Mourad, and M.~Juntti,
  ``Reconfigurable intelligent surface for 5g {NR} uplink coverage
  enhancement,'' in \emph{2021 IEEE 94th Vehicular Technology Conference
  {(VTC2021-Fall)}}, 2021, pp. 1--5.

\bibitem{Vincenzi:17}
M.~Vincenzi, A.~Antonopoulos, E.~Kartsakli, J.~Vardakas, L.~Alonso, and
  C.~Verikoukis, ``Cooperation incentives for multi-operator {C-RAN} energy
  efficient sharing,'' in \emph{2017 IEEE International Conference on
  Communications (ICC)}, 2017, pp. 1--6.

\bibitem{Chien:18}
H.-T. Chien, Y.-D. Lin, H.-W. Chang, and C.-L. Lai, ``Multi-operator fairness
  in transparent {RAN} sharing by soft-partition with blocking and dropping
  mechanism,'' \emph{IEEE Transactions on Vehicular Technology}, vol.~67,
  no.~12, pp. 11\,597--11\,605, 2018.

\bibitem{kayestimation}
S.~M. Kay, \emph{Fundamentals of Statistical Signal Processing, Volume I:
  Estimation Theory}.\hskip 1em plus 0.5em minus 0.4em\relax Prentice Hall,
  1993.

\bibitem{poorbook}
H.~V. Poor, \emph{An Introduction to Signal Detection and Estimation (2nd
  Ed.)}.\hskip 1em plus 0.5em minus 0.4em\relax Berlin, Heidelberg:
  Springer-Verlag, 1994.

\bibitem{MarzettaBook}
T.~L. Marzetta, E.~G. Larsson, H.~Yang, and H.~Q. Ngo, \emph{Fundamentals of
  Massive {MIMO}}.\hskip 1em plus 0.5em minus 0.4em\relax Cambridge University
  Press, 2016.

\bibitem{Bjornson2021spatial}
E.~Bj\"ornson and L.~Sanguinetti, ``Rayleigh fading modeling and channel
  hardening for reconfigurable intelligent surfaces,'' \emph{IEEE Wireless
  Communications Letters}, vol.~10, no.~4, pp. 830--834, 2021.

\end{thebibliography}
\end{document}